\documentclass[final]{siamltex}
\usepackage{amssymb,amsmath}
\usepackage[numbered]{mcode}
\usepackage{graphicx,url}
\graphicspath{{figures/}}
\usepackage{clrs+} % for id,func
\usepackage{algorithm}
\usepackage{subfig}
\usepackage[noend]{algpseudocode}
\usepackage{enumitem}
\setlist{noitemsep,leftmargin=*}

\graphicspath{{figures/}}

\newcommand{\dimN}{n}
\newcommand{\dimM}{m}
\newcommand{\dimK}{k}
\newcommand{\dnzc}{\id{nzc}}
\newcommand{\dnzr}{\id{nzr}}
\newcommand{\dni}{\id{ni}}
\newcommand{\dnnz}{\id{nnz}}

\newcommand{\dmax}{\func{max}}
\newcommand{\dth}{th}
\newcommand{\dlen}{\id{len}}
\newcommand{\NUM}{\mathsf{NUM}}
\newcommand{\IR}{\mathsf{IR}}
\newcommand{\JC}{\mathsf{JC}}
\newcommand{\CP}{\mathsf{CP}}
\newcommand{\AUX}{\mathsf{AUX}}

\newcommand{\qC}{\c{C}}

\newcommand{\matlab}{{\sc Matlab}}
\newcommand{\mA}{\mathbf{A}} 
\newcommand{\mB}{\mathbf{B}}
\newcommand{\transpose}     {^{\mbox{\scriptsize \sf T}}}

\newcommand{\mC}{\mathbf{C}}
\newcommand{\flops}{\mathrm{flops}}
\newcommand{\rmat}{R-MAT}
\newcommand{\erdosrenyi}{Erd\H os-R\'{e}nyi}

\def\Cpp{C{}\texttt{++}~}

\newcommand{\lilabel}[1]        {\label{li:#1}}
\newcommand{\liref}[1]      {line~\ref{li:#1}}

\newcommand{\lirefs}[2]     {lines \ref{li:#1}--\ref{li:#2}}

\title{ Parallel Sparse Matrix-Matrix Multiplication and Indexing: 
	Implementation and Experiments 
        \thanks{This work was supported in part by NSF grant CNS-0709385, and by grants from Intel Corporation and Microsoft Corporation.
        All authors from LBNL were supported by the ASCR Office in the DOE Office of Science under contract number DE-AC02-05CH11231.}}

\author{Aydin Bulu\qC\ \thanks{Computational Research Division, Lawrence Berkeley National Laboratory,
1 Cyclotron Road, Berkeley, CA 94720 ({\tt abuluc@lbl.gov}).}
        \and John R. Gilbert \thanks{Computer Science Department,
University of California,
Santa Barbara, CA 93106-5110 ({\tt gilbert@cs.ucsb.edu}).}}

\begin{document}

\maketitle

\begin{abstract}
Generalized sparse matrix-matrix multiplication (or SpGEMM) is a
key primitive for many high performance graph algorithms as well 
as for some linear solvers, such as algebraic multigrid. 
Here we show that SpGEMM also yields efficient algorithms for 
general sparse-matrix indexing in distributed memory, 
provided that the underlying SpGEMM implementation
is sufficiently flexible and scalable.
We demonstrate that our parallel SpGEMM methods,
which use two-dimensional block data distributions with
serial hypersparse kernels, 
are indeed highly flexible, scalable, and memory-efficient
in the general case.
This algorithm is the first to yield increasing speedup on an 
unbounded number of processors;
our experiments show scaling up to thousands of processors in a
variety of test scenarios. 
\end{abstract}



\pagestyle{myheadings}
\thispagestyle{plain}
\markboth{BULU\qC\ AND GILBERT}{SPARSE MATRIX INDEXING AND MULTIPLICATION}

\section{Introduction}
We describe scalable parallel implementations of two sparse matrix kernels.
The first, SpGEMM, computes the product of two sparse matrices over a general semiring.
The second, SpRef, performs generalized indexing into a sparse matrix: 
Given vectors $\mathsf{I}$ and $\mathsf{J}$ of row and column indices, 
SpRef extracts the submatrix $\mA(\mathsf{I},\mathsf{J})$.  
Our novel approach to SpRef uses SpGEMM as its key subroutine, 
which regularizes the computation and data access patterns;
conversely, applying SpGEMM to SpRef emphasizes the importance
of an SpGEMM implementation that handles arbitrary matrix
shapes and sparsity patterns, and a complexity analysis
that applies to the general case.

Our main contributions in this paper are: 
first, we show that SpGEMM leads to a simple and efficient implementation
of SpRef; second, we describe a distributed-memory implementation of SpGEMM 
that is more general in application and more flexible in processor 
layout than before; and, third, we report on extensive experiments with the performance 
of SpGEMM and SpRef. We also describe an algorithm for sparse matrix assignment (SpAsgn), 
and report its parallel performance.  The SpAsgn operation, formally $\mA(\mathsf{I},\mathsf{J}) = \mB$, assigns a sparse matrix to a submatrix of another sparse matrix.
It can be used to perform streaming batch updates to a graph. 

Parallel algorithms for SpGEMM and SpRef, as well as their theoretical performance, are described in Sections~\ref{sec:spgemm} and~\ref{sec:spref}. 
We present the general SpGEMM algorithm and its parallel complexity before SpRef since the latter uses SpGEMM as a subroutine and
its analysis uses results from the SpGEMM analysis. 
Section~\ref{sec:parspgemm} summarizes our earlier results on the complexity of various SpGEMM algorithms on distributed memory. 
Section~\ref{sec:sparsesumma} presents our algorithm of choice, Sparse SUMMA, in a more formal way than before, including a pseudocode 
general enough to handle different blocking parameters, rectangular matrices, and rectangular processor grids. 
The reader interested only in parallel SpRef can skip these sections and go directly to Section~\ref{sec:spref}, where we describe our SpRef algorithm, 
its novel parallelization and its analysis. %, if he or she is willing to take the SpGEMM dependencies on faith.
Section~\ref{sec:expparallel} gives an extensive performance evaluation of these two primitives using 
large scale parallel experiments, including a performance comparison with similar primitives from the Trilinos package. 
Various implementation decisions and their effects on performance are also detailed.

\section{Notation}
Let $\mA \in \mathbb{S}^{\dimM \times \dimN}$ be a sparse rectangular matrix of elements from a semiring $\mathbb{S}$. 
We use $\dnnz(\mA)$ to denote the number of nonzero elements in $\mA$. When the
matrix is clear from context, we drop the parenthesis and simply use $\dnnz$. For sparse matrix indexing, we use the convenient 
${\textrm \matlab }$ colon notation, where $\mA(:,i)$ denotes the $i$\dth\ column, $\mA(i,:)$ denotes  the $i$\dth\ row, and $\mA(i,j)$ denotes the 
element at the $(i,j)$\dth\ position of matrix $\mA$. Array and vector indices are 1-based throughout this paper. 
The length of an array $\mathsf{I}$,  denoted by $\dlen(\mathsf{I})$, is equal to its number of elements.
For one-dimensional arrays, $\mathsf{I}(i)$ denotes the $i$\dth\ component of the array.
We use $\flops(\mA\, \cdot\, \mB)$, pronounced ``flops'', 
to denote the number of nonzero arithmetic operations required when computing the product of matrices $\mA$ and $\mB$. 
Since the flops required to form the matrix triple product differ depending on the order of multiplication, $\flops((\mA \mB) \cdot \mC)$ and 
$\flops(\mA \cdot (\mB \mC))$ mean different things. The former is the flops needed to multiply the product $\mA \mB$ with $\mC$, where 
the latter is the flops needed to multiply $\mA$ with the product $\mB \mC$. 
When the operation and the operands are clear from context, we simply use $\flops$. 
The \matlab\ {\tt sparse(i,j,v,m,n)} function, which is used in some of the pseudocode, creates an $m\times n$ sparse matrix $\mA$ with nonzeros $\mA(i(k), j(k)) = v(k)$.

In our analyses of parallel running
time, the latency of sending a message over the communication interconnect is $\alpha $, 
and the inverse bandwidth is $\beta $, both expressed in terms of time for a floating-point 
operation (also accounting for the cost of cache misses and memory indirections associated with that floating point operation).
$f(x) = \Theta(g(x))$ means that $f$ is bounded asymptotically by $g$ both above and below.

\section{Sparse matrix-matrix multiplication}
\label{sec:spgemm}

\begin{figure}[t]
\centering  
\includegraphics[scale=0.45]{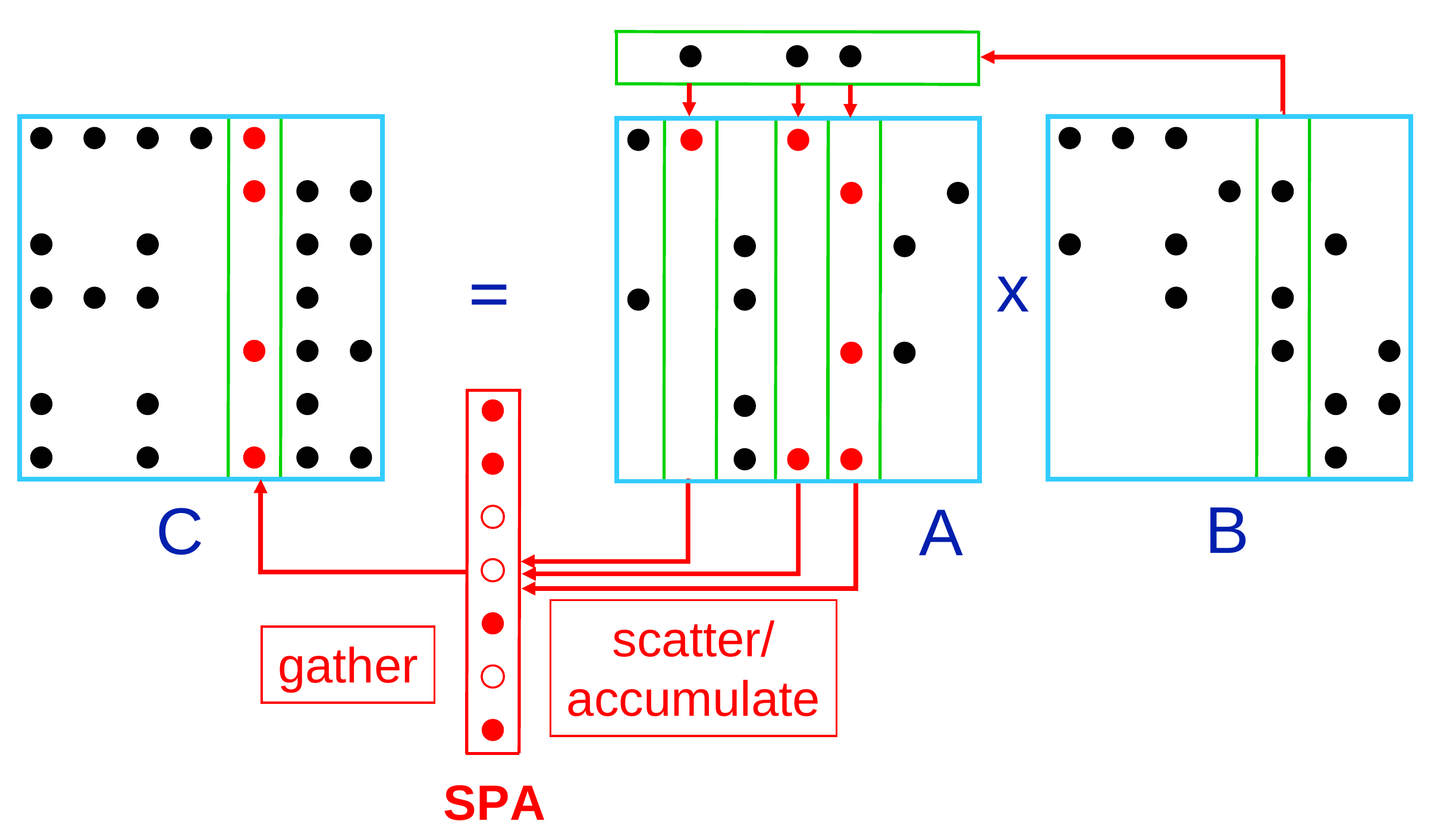}  
\caption[Multiply sparse matrices column-by-column]
{Multiplication of sparse matrices stored by columns~\cite{GALLA/sparse}. Columns of $\mA$ are accumulated as specified by the non-zero entries in a column of $\mB$ using a sparse accumulator or SPA~\cite{smatlab}.  
The contents of the SPA are stored into a column of $\mC$ once all required columns are accumulated. } 
\label{fig:cscmultpic}
\end{figure}

SpGEMM is a building block for many high-performance graph algorithms, including graph 
contraction~\cite{unifiedstarp}, breadth-first search from multiple source vertices~\cite{combblas}, 
peer pressure clustering~\cite{shahthesis}, recursive all-pairs shortest-paths~\cite{rkleene}, 
matching~\cite{matching}, and cycle detection~\cite{cycle}. It is a subroutine in more traditional 
scientific computing applications such as multigrid interpolation and restriction~\cite{multigrid} 
and Schur complement methods in hybrid linear solvers~\cite{YamazakiSchurComp10}. 
It also has applications in general computing, including parsing 
context-free languages~\cite{sparseclosure} and colored intersection searching~\cite{Kaplan:2006}.

The classical serial SpGEMM algorithm for general sparse matrices was first described by Gustavson~\cite{gust:78}, 
and was subsequently used in Matlab~\cite{smatlab} and CSparse~\cite{davisbook}. 
That algorithm, shown in Figure~\ref{fig:cscmultpic},
runs in $O(\flops + \dnnz+\dimN)$ time, which is optimal for $\flops \geq \dmax\{\dnnz,\dimN\}$. It uses the popular 
compressed sparse column (CSC) format for representing its sparse matrices. Algorithm~\ref{alg:cscgemm} gives the pseudocode for this column-wise serial algorithm for SpGEMM. 

\begin{algorithm}
\begin{algorithmic}[1]
\Procedure{Columnwise-SpGEMM}{$\mA, \mB, \mC$}
\For{$j \gets 1$ to $\dimN$}
	\For{$k$ where $\mB(k,j) \neq 0 $}	
	\State $ \mC(:, j) \gets \mC(:, j) + \mA(:, k) \cdot \mB(k, j) $
	\EndFor
\EndFor
\EndProcedure
\end{algorithmic}
\caption{Column-wise formulation of serial matrix multiplication} \label{alg:cscgemm}
\end{algorithm}

\subsection{Distributed memory SpGEMM}
\label{sec:parspgemm}
The first question for distributed memory algorithms is ``where is the data?''. 
In parallel SpGEMM, we consider two ways of distributing data to processors. 
In 1D algorithms, each processor stores a block of $\dimM/p$ rows of an $\dimM$-by-$\dimN$ sparse matrix. 
In 2D algorithms, processors are logically organized as a rectangular $p = p_r \times p_c$ grid, 
so that a typical processor is named $P(i,j)$. 
Submatrices are assigned to processors according to a 2D block decomposition:
processor $P(i,j)$ stores the submatrix $\mA_{ij}$ of dimensions $(m/p_r)\times (n/p_c)$ in its local memory. 
We extend the colon notation to slices of submatrices:  $\mA_{i:}$ denotes the  $(m/p_r)\times n$ slice of $\mA$
collectively owned by all the processors along the $i$th processor row and $\mA_{:j}$ denotes the  $m\times (n/p_c)$ slice of $\mA$
collectively owned by all the processors along the $j$th processor column.

We have previously shown that known 1D SpGEMM algorithms are not scalable to thousands of processors~\cite{icpp08}, 
while 2D algorithms can potentially speed up indefinitely, albeit with decreasing efficiency. 
There are two reasons that the 1D algorithms do not scale:
First, their auxiliary data structures cannot be loaded and unloaded fast enough to amortize their costs. 
This loading and unloading is necessary because the 1D algorithms proceed in stages in which only 
one processor broadcasts its submatrix to the others, in order to avoid running out of memory.    
Second, and more fundamentally, the communication costs of 1D algorithms are not scalable regardless of data structures.
Each processor receives $\dnnz$ (of either $\mA$ or $\mB$) data in the worst case, 
which implies that communication cost is on the same order as computation,
prohibiting speedup beyond a fixed number of processors.
This leaves us with 2D algorithms for a scalable solution.

Our previous work~\cite{ipdps08} shows that the standard compressed column or row 
(CSC or CSR) data structures
are too wasteful for storing the local submatrices arising from a 2D decomposition.
This is because the local submatrices are {\em hypersparse}, 
meaning that the ratio of nonzeros to dimension is asymptotically zero.
The total memory across all processors for CSC format would be $O(\dimN\sqrt{p}+\dnnz)$, 
as opposed to $O(\dimN+\dnnz)$ memory to store the whole matrix in CSC on a single processor. 
Thus a scalable parallel 2D data structure must respect hypersparsity.

Similarly, any algorithm whose complexity depends on matrix dimension, such as Gustavson's serial SpGEMM algorithm, 
is asymptotically too wasteful to be used as a computational kernel for multiplying the hypersparse submatrices.
% at \protect{\liref{localmult}} in Algorithm~\ref{alg:sparsesumma}. 
Our \proc{HyperSparseGEMM}~\cite{GALLA/spgemm, ipdps08}, on the other hand, operates on the strictly $O(\dnnz)$ 
{\em doubly compressed sparse column (DCSC)} data structure, 
and its time complexity does not depend on the matrix dimension. 
Section~\ref{sec:dcsc} gives a succinct summary of DCSC.
 
Our HyperSparseGEMM uses an outer-product formulation whose time complexity is
$O(\dnzc(\mA)+\dnzr(\mB)+ \flops \cdot \lg{\dni})$, 
where $\dnzc(\mA)$ is the number of columns of $\mA$ that are not entirely zero, 
$\dnzr(\mB)$ is the number of rows of $\mB$ that are not entirely zero,
and $\dni$ is the number of indices $i$ for which $\mA(:,i)\not= \emptyset$ and $\mB(i,:)\not= \emptyset$.
The extra $\lg{\dni}$ factor at the time complexity originates from the priority queue that is used to merge $\dni$ outer products 
on the fly. The overall memory requirement of this algorithm is the asymptotically optimal
$O(\dnnz(\mA) + \dnnz(\mB) + \dnnz(\mC))$, independent of either matrix dimensions or $\flops$.

\subsection{DCSC Data Structure}
\label{sec:dcsc}

\begin{figure}
\begin{center}
${\scriptstyle
	\left|
	\begin {array}{rrrrrrrrrrrrrrrr}  
	\CP	& = 	& 1 		& 3	& 3	& 3 	&  3	& 3	&  3		& 4 		&  5	& 5			\\\noalign{\smallskip} 
	    	&   	& \downarrow	& 	& 	& 	& 	& 	&  \downarrow	& \downarrow 	&  	& 		 	\\\noalign{\smallskip} 
	\IR 	& = 	& 6		& 8	& 	&  	& 	& 	&  4		& 2 		&  	& 			\\\noalign{\smallskip} 
	\NUM    & = 	& 0.1		& 0.2	& 	& 	&	&  	&  0.3		&  0.4		& 	&  			\\\noalign{\smallskip} 
	\end {array}  \right| 
}$
\end{center}
\caption{Matrix $\mA$ in CSC format}
\label{fig:csc}
\end{figure}

DCSC~\cite{ipdps08} is a further compressed version of CSC where repetitions in the column pointers array, which arise from empty columns, are not allowed. 
Only columns that have at least one nonzero are represented, together with their column indices.  

\begin{figure}[hb]
  \begin{center}
   \vspace*{1ex}
  \begin{minipage}[b]{0.5\linewidth}
       \begin{center}
	\raisebox{-0.6cm}{\includegraphics{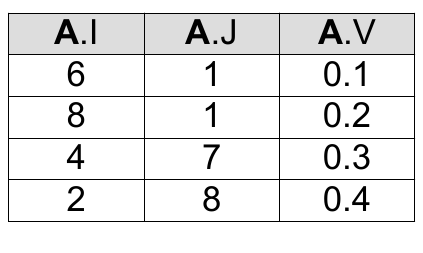}}
      \newline
     \caption{Matrix $\mA$ in Triples format \label{fig:triples}}
       \end{center}
    \end{minipage}\hfill
    \begin{minipage}[b]{0.5\linewidth}
       \begin{center}
      \includegraphics[scale=0.9]{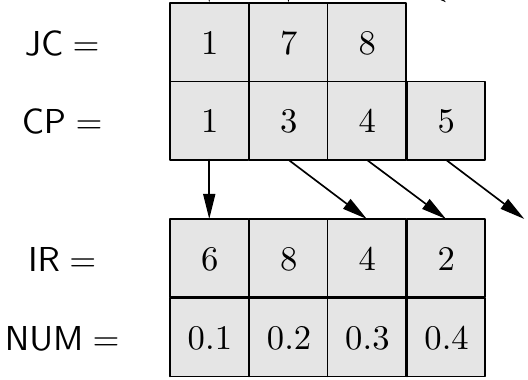}
     \caption{Matrix $\mA$  in DCSC format  \label{fig:dcsc}}
       \end{center}
    \end{minipage}
  \end{center}
\end{figure}

For example, consider the 9-by-9 matrix with 4 nonzeros as in Figure~\ref{fig:triples}.
Figure~\ref{fig:csc} showns its CSC storage, which includes repetitions and redundancies in the column pointers array ($\CP$). 
Our new data structure compresses this column pointers array to avoid repetitions, giving $\CP$ of DCSC as 
in Figure~\ref{fig:dcsc}. DCSC is essentially a sparse array of sparse columns, whereas CSC is a dense array of sparse columns.

After removing repetitions, $\CP(i)$ does no longer refer to the $i$\dth\ column. 
A new $\JC$ array, which is parallel to $\CP$, gives us the column numbers. 
Although our \proc{Hypersparse\_GEMM} algorithm does not need column indexing, DCSC can support fast column indexing by building 
an $\AUX$ array that contains pointers to nonzero columns (columns that have at least one nonzero element) in linear time.

\subsection{Sparse SUMMA algorithm}
\label{sec:sparsesumma}
Our parallel algorithm is inspired by the dense matrix-matrix multiplication algorithm SUMMA~\cite{summa}, 
used in parallel BLAS~\cite{parallelblas}.
SUMMA is memory efficient and easy to generalize to non-square matrices and processor grids.
 
The pseudocode of our 2D algorithm, \proc{SparseSUMMA}~\cite{icpp08}, is shown in Algorithm~\ref{alg:sparsesumma} in its most general form. The coarseness 
of the algorithm can be adjusted by changing the block size $1 \leq b \leq \gcd(\dimK/p_r, \dimK/p_c)$. For the first time, we present the algorithm in a 
form general enough to handle rectangular processor grids and a wide range of blocking parameter choices. The pseudocode, however, requires $b$ to evenly divide 
$\dimK/p_r$ and $\dimK/p_c$ for ease of presentation. This requirement can be dropped at the expense of having potentially multiple broadcasters 
along a given processor row and column during one iteration of the loop starting at \liref{stage}. The $\textbf{for}\ldots\textbf{in parallel do}$ construct indicates
that all of the \textbf{do} code blocks execute in parallel by all the processors. The execution of the algorithm on a rectangular grid with rectangular
sparse matrices is illustrated in Figure~\ref{fig:spsumma}. We refer to the Combinatorial BLAS source code~\cite{combblas_website} for additional details.
 
\begin{algorithm}
\begin{algorithmic}[1]
\Require $\mA \in \mathbb{S}^{\dimM \times \dimK},\mB \in \mathbb{S}^{\dimK \times \dimN}$: sparse matrices distributed on a $p_r \times p_c$ processor grid
\Ensure $\mC \in \mathbb{S}^{\dimM \times \dimN}$: the product $\mA \mB$, similarly distributed. 
\Procedure{SparseSUMMA}{$\mA, \mB, \mC$}
\For{all processors $P(i,j)$\  \InParallel} 
\State $ \mB_{ij} \gets (\mB_{ij})\transpose$ \lilabel{transpose}
%\EndFor
%\For{all processors $P(i,j)$\  \InParallel} 
	\For{ $q =1$ to $\dimK/b$} \Comment{blocking parameter $b$ evenly divides $\dimK/p_r$ and $\dimK/p_c$} \lilabel{stage}
	\State $c = (q \cdot b)/p_c$ \Comment{$c$ is the broadcasting processor column}
	\State $r = (q \cdot b)/p_r$ \Comment{$r$ is the broadcasting processor row}
 	\State $\id{lcols} = (q \cdot b) \bmod{p_c}:((q+1) \cdot b) \bmod p_c$ \Comment{local column range} \lilabel{loccol}
	\State $\id{lrows} = (q \cdot b) \bmod{p_r}:((q+1) \cdot b) \bmod p_r$ \Comment{local row range} \lilabel{locrow}
	\State	$\mA^{rem} \gets $ \Call{Broadcast}{$\mA_{ic}(:,\id{lcols}), P(i,:)$} 
	\State	$\mB^{rem} \gets $ \Call{Broadcast}{$\mB_{rj}(:,\id{lrows}), P(:,j)$} 		
	\State	$\mC_{ij} \gets \mC_{ij} + \Call{HyperSparseGEMM}{\mA^{rem},\mB^{rem}}$	\lilabel{localmult}
	\EndFor
%\EndFor
%\For{all processors $P(i,j)$\  \InParallel} 
\State $ \mB_{ij} \gets (\mB_{ij})\transpose$  \Comment{Restore the original $\mB$}
\EndFor
\EndProcedure
\end{algorithmic}
\caption{Operation $ \mC \gets \mA  \mB$ using Sparse SUMMA} \label{alg:sparsesumma}
\end{algorithm}

\begin{figure}
\begin{center}
\includegraphics[scale=0.6]{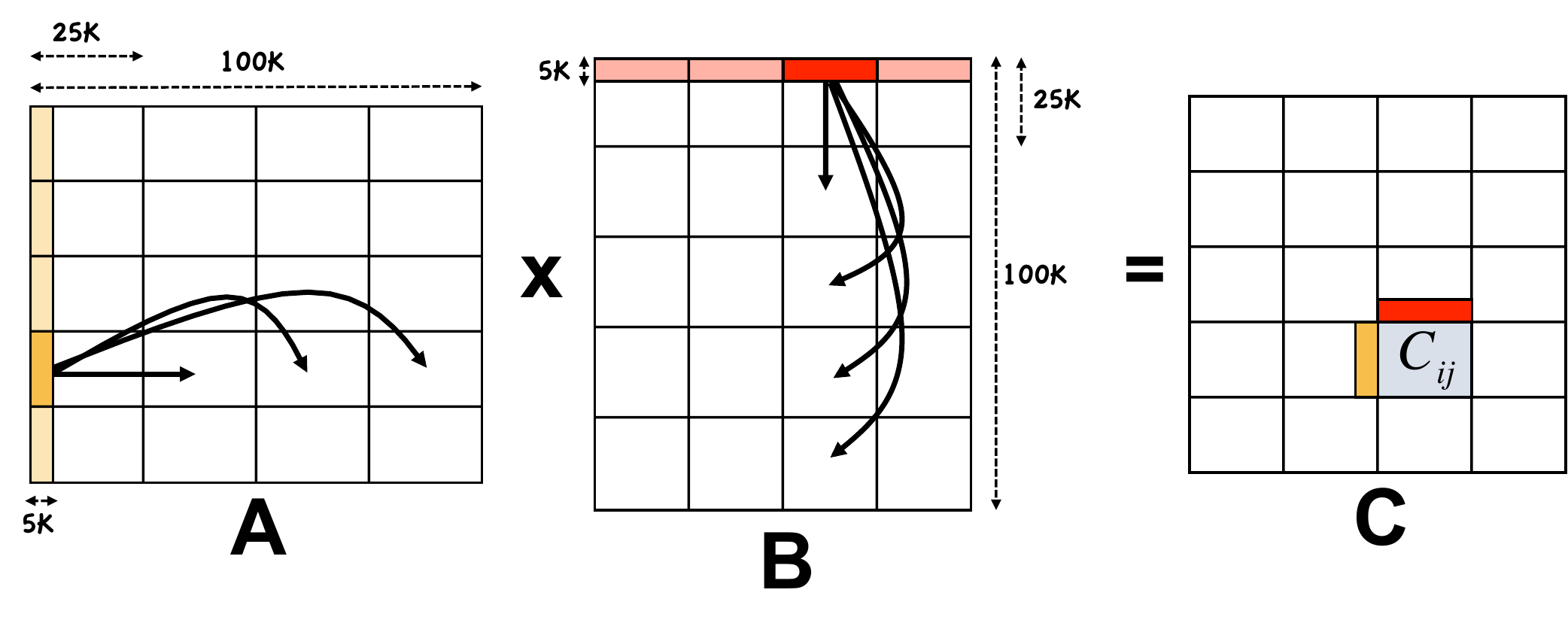}
\end{center}
\caption{ Execution of the Sparse SUMMA algorithm for sparse matrix-matrix multiplication $\mC = \mA \cdot \mB$. The example shows the first stage of the algorithm execution 
(the broadcast and the local update by processor $P(i,j)$). The two rectangular sparse operands $\mA$ and $\mB$ are of sizes $\dimM$-by-$100K$ and 
$100K$-by-$\dimN$, distributed on a $5 \times 4$ processor grid. Block size $b$ is chosen to be $5K$. 
\label{fig:spsumma}}
\end{figure} 

The \proc{Broadcast}($\mA_{ic}, P(i,:)$) syntax means that the owner of $\mA_{ic}$ becomes the root and broadcasts
its submatrix to all the processors on the $i$th processor row. Similarly for \proc{Broadcast}($\mB_{rj}, P(:,j)$),
the owner of $\mB_{rj}$ broadcasts its submatrix to all the processors on the $j$th processor column. 
In \lirefs{loccol}{locrow}, we find the local column (for $\mA$) and row (for $\mB$) ranges for matrices that are 
to be broadcast during that iteration. They are significant only at the broadcasting processors, which can be determined 
implicitly from the first parameter of \proc{Broadcast}. We index $\mB$ by columns as opposed to rows because it
has already been locally transposed in \liref{transpose}. This makes indexing faster since local submatrices are stored in 
the column-based DCSC sparse data structure.  
Using DCSC, the expected cost of fetching 
$b$ consecutive columns of a matrix $\mA$ is $b$ plus the size (number of nonzeros) of the output.
Therefore, the algorithm asymptotically has the same computation cost for all values of $b$.

For our complexity analysis, we assume that nonzeros of input sparse matrices are independently and identically distributed, input matrices
are $\dimN$-by-$\dimN$, with $d > 0$ nonzeros per row and column on the average.  The sparsity parameter $d$ simplifies our analysis by
making different terms in the complexity comparable to each other. For example, if $\mA$ and $\mB$ both have 
sparsity $d$, then $\dnnz(\mA) = d \dimN$ and $\flops(\mA \mB) = d^2 \dimN$.

The communication cost of the Sparse SUMMA algorithm, for the case of $p_r=p_c=\sqrt{p}$, is
\begin{equation} 
T_{comm} 	= \sqrt{p} \: \Bigl( 2\, \alpha + \beta\, \bigl( \frac{\dnnz(\mA) +\dnnz(\mB)}{p} \bigr) \Bigr) 
		= \Theta (\alpha\, \sqrt{p} + \frac{\beta\, d\, \dimN}{\sqrt{p}}),
\label{eqn:tcomm2d}
\end{equation}
and its computation cost is
\begin{equation}
T_{comp} = O \Bigl( \frac{d\, n}{\sqrt{p}}  +  \frac{d^2  n}{p}  \lg \bigl(  \frac{d^2 n}{p \sqrt{p}}  \bigr) + \frac{d^2 n \lg{\sqrt{p}}}{p} \Bigr) = O  \Bigl( \frac{d\, n}{\sqrt{p}}  +  \frac{d^2  n}{p}   \lg \bigl(  \frac{d^2 n}{p}  \bigr)  \Bigr)~\cite{GALLA/spgemm}.
\label{eqn:tcomp2d}	
\end{equation}

We see that although scalability is not perfect and efficiency deteriorates as $p$ increases,  the achievable speedup is not bounded.
Since $\lg(d^2 n/ p)$ becomes negligible as $p$ increases, the bottlenecks for scalability are the $\beta \, d \, n/\sqrt{p}$ term of $T_{comm}$ and 
the $d\, n/\sqrt{p}$ term of $T_{comp}$, which scale with $\sqrt{p}$. Consequently, two different scaling regimes are likely to be present: 
A close to linear scaling regime until those terms start to dominate and a $\sqrt{p}$-scaling regime afterwards.

\section{Sparse matrix indexing and subgraph selection}
\label{sec:spref}
Given a sparse matrix $\mA$ and two vectors $\mathsf{I}$ and $\mathsf{J}$ of indices,
SpRef extracts a submatrix and stores it as another sparse matrix, $\mB = \mA(\mathsf{I},\mathsf{J})$. 
Matrix $\mB$ contains the elements in rows $\mathsf{I}(i)$
and columns $\mathsf{J}(j)$ of $\mA$,
for $i=1,...,\dlen(\mathsf{I})$ and $j=1,...,\dlen(\mathsf{J})$, respecting the order of indices. 
If $\mA$ is the adjacency matrix of a graph, SpRef($\mA,\mathsf{I}$, $\mathsf{I})$ selects an induced subgraph. 
SpRef can also be used to randomly permute the rows and columns of a sparse matrix,
a primitive in parallel matrix computations commonly used for load balancing~\cite{ogielski}. 
 
Simple cases such as row ($\mA(i,:)$), column ($\mA(:,i)$), and element ($\mA(i,j)$) indexing are often handled 
by special purpose subroutines~\cite{GALLA/sparse}. A parallel algorithm for the general case, where $\mathsf{I}$ and 
$\mathsf{J}$ are arbitrary vectors of indices, does not exist in the literature. We propose an algorithm that uses parallel 
SpGEMM. Our algorithm is amenable to performance analysis for the general case.

A related kernel is SpAsgn, or sparse matrix assignment. 
This operation assigns a sparse matrix to a submatrix of another sparse matrix,
$\mA(\mathsf{I},\mathsf{J}) = \mB$.
A variation of SpAsgn is $\mA(\mathsf{I},\mathsf{J}) = \mA(\mathsf{I},\mathsf{J}) + \mB$, 
which is similar to Liu's {\it extend-add} operation~\cite{multifrontal}
in finite element matrix assembly. 
Here we describe the sequential SpAsgn algorithm and its analysis, 
and report large-scale performance results in Section~\ref{sec:exp:spasgn}.

\begin{figure}[h]
\begin{center}
\includegraphics[scale=0.55]{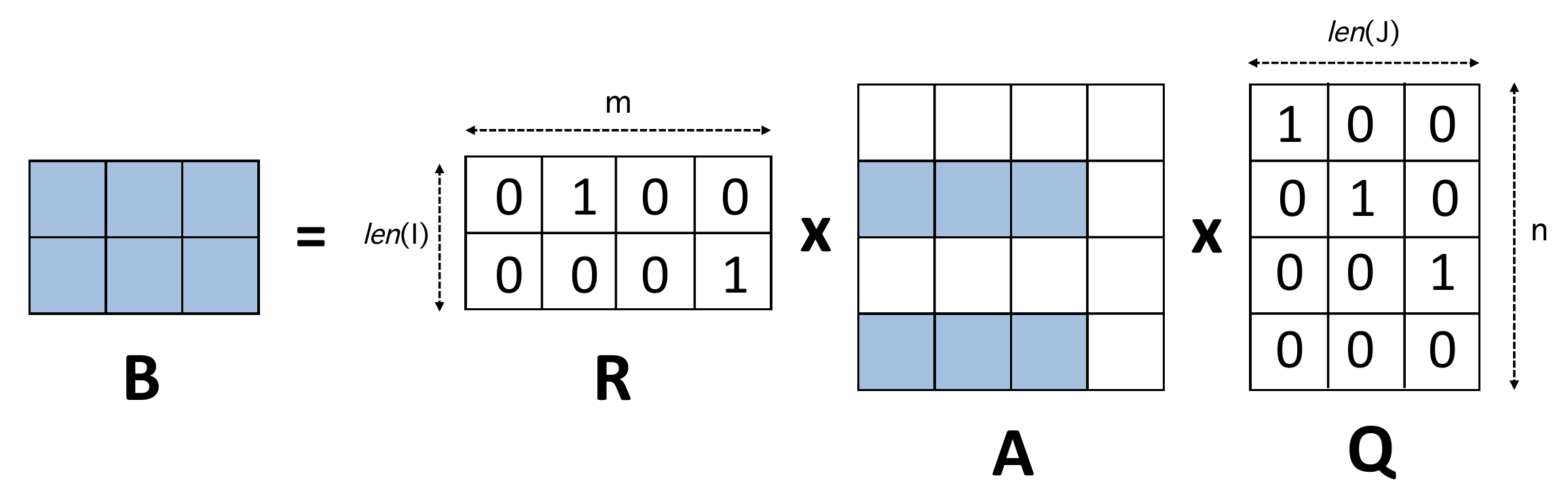}
\end{center}
\caption{ Sparse matrix indexing (SpRef) using mixed-mode SpGEMM. On an $m$-by-$n$ matrix $\mA$, 
the SpRef operation $\mA(\mathsf{I},\mathsf{J})$ extracts a $\dlen(\mathsf{I})$-by-$\dlen(\mathsf{J})$ submatrix, 
where $\mathsf{I}$ is a vector of row indices and $\mathsf{J}$ is a vector of column indices. 
The example shows $\mB = \mA([2,4],[1,2,3])$. 
It performs two SpGEMM operations between a boolean matrix and a general-type matrix. 
\label{fig:spref}}
\end{figure} 

\subsection{Sequential algorithms for SpRef and SpAsgn}
Performing SpRef by a triple sparse-matrix product is illustrated in Figure~\ref{fig:spref}. The algorithm can be described concisely in 
Matlab notation as follows:

\begin{tabular}{c}
\begin{lstlisting}
function B = spref(A,I,J)

[m,n] = size(A);
R = sparse(1:len(I),I,1,len(I),m);
Q = sparse(J,1:len(J),1,n,len(J));
B = R*A*Q;
\end{lstlisting}
\end{tabular}

The sequential complexity of this algorithm is $\flops(\mathbf{R} \cdot \mA) + \flops((\mathbf{R} \mA) \cdot \mathbf{Q})$. 
Due to the special structure of the permutation matrices, 
the number of nonzero operations required to form the product $\mathbf{R} \cdot \mA$ is equal to 
the number of nonzero elements in the product. 
That is, $\flops(\mathbf{R}\cdot \mA) = \dnnz(\mathbf{R}\mA) \leq \dnnz(\mA)$. 
Similarly, $\flops((\mathbf{R} \mA) \cdot \mathbf{Q}) \leq \id{nnz}(\mA)$, 
making the overall complexity $O(\dnnz(\mA))$ 
for any $\mathsf{I}$ and $\mathsf{J}$. 
This is optimal in general, 
since just writing down the result of a matrix permutation
$\mB = \mA(r,r)$ requires $\Omega( \dnnz(\mA))$ operations. 

Performing SpAsgn by two triple sparse-matrix products and additions is illustrated in Figure~\ref{fig:spasgn}.
We create two temporary sparse matrices of the same dimensions as $\mA$. These matrices contain nonzeros 
only for the $\mA(\mathsf{I},\mathsf{J})$ part, and zeros elsewhere. The first triple product embeds $\mB$ into a 
bigger sparse matrix that we add to $\mA$. The second triple product embeds $\mA(\mathsf{I},\mathsf{J})$ into 
an identically sized sparse matrix so that we can zero out the $\mA(\mathsf{I},\mathsf{J})$ portion by subtracting 
it from $\mA$. Since general semiring axioms do not require additive inverses to exist, we implement this piece of the algorithm
slightly differently that stated in the pseudocode. We still form the $\mathbf{S}\mA\mathbf{T}$ product but instead of using subtraction, 
we use the generalized sparse elementwise multiplication function of the Combinatorial BLAS~\cite{combblas}
to zero out the $\mA(\mathsf{I},\mathsf{J})$ portion. In particular, we first perform an 
elementwise multiplication of $\mA$ with the negation of $\mathbf{S}\mA\mathbf{T}$ without explicitly forming the negated matrix, 
which can be dense. Thanks to this direct support for the implicit negation operation, the complexity bounds are identical to the version
that uses subtraction.  The negation does not assume additive inverses: it sets all zero entries to one and all nonzeros entries to zero.
The algorithm can be described concisely in Matlab notation as follows:
  
\begin{tabular}{c}
\begin{lstlisting}
function C = spasgn(A,I,J,B)
% A = spasgn(A,I,J,B) performs A(I,J) = B

[ma,na] = size(A);
[mb,nb] = size(B);
R = sparse(I,1:mb,1,ma,mb);
Q = sparse(1:nb,J,1,nb,na);
S = sparse(I,I,1,ma,ma);
T = sparse(J,J,1,na,na);
C = A  + R*B*Q - S*A*T;
\end{lstlisting}
\end{tabular}

Liu's {\it extend-add} operation is similar to SpAsgn but simpler; it just omits subtracting the $\mathbf{S} \mA \mathbf{T}$ term. 

\begin{figure}
\[
\mA  = \mA + \left(\begin{array}{ccc}
     0 		& 0	  & 0	 		\\
     0		& \mB     &  0 		 	\\
     0 		& 0	  & 0	 \end{array}\right) 
- \left(\begin{array}{ccc}
     0 		& 0	  & 0	 		\\
     0		& \mA(\mathsf{I},\mathsf{J})  & 0 		 	\\
     0 		& 0	  & 0	 		
\end{array}\right) 
\]
\caption{ Illustration of SpAsgn ($\mA(\mathsf{I},\mathsf{J}) = \mB$) for rings where additive inverses are defined. 
For simplicity, the vector indices $\mathsf{I}$ and
$\mathsf{J}$ are shown as contiguous, but they need not be.
\label{fig:spasgn}}
\end{figure}

Let us analyze the complexity of SpAsgn. Given $\mA \in \mathbb{S}^{\dimM \times \dimN}$ and 
$\mB \in \mathbb{S}^{\dlen(\mathsf{I}) \times \dlen(\mathsf{J})}$, the intermediate boolean matrices have the following properties:

$\mathbf{R}$ is $\dimM$-by-$\dlen(\mathsf{I})$ rectangular with $\dlen(\mathsf{I})$ nonzeros, one in each column.

$\mathbf{Q}$ is $\dlen(\mathsf{J})$-by-$\dimN$ rectangular with $\dlen(\mathsf{J})$ nonzeros, one in each row.  

$\mathbf{S}$ is $\dimM$-by-$\dimM$ symmetric with $\dlen(\mathsf{I})$ nonzeros, all located along the diagonal.  

$\mathbf{T}$ is $\dimN$-by-$\dimN$ symmetric with $\dlen(\mathsf{J})$ nonzeros, all located along the diagonal.    

\begin{theorem}
The sequential SpAsgn algorithm takes $O(\dnnz(\mA)+\dnnz(\mB)+\dlen(\mathsf{I})+\dlen(\mathsf{J}))$ time using an optimal 
 $\Theta(\flops)$ SpGEMM subroutine. 
\end{theorem}

\begin{proof}
The product $\mathbf{R} \cdot \mB$ requires $\flops(\mathbf{R} \cdot \mB) = \dnnz( \mathbf{R} \mB) = \dnnz(\mB)$ operations because
there is a one-to-one relationship between nonzeros in the output and $\flops$ performed. Similarly, 
$\flops((\mathbf{R} \mB) \cdot \mathbf{Q}) = \dnnz( \mathbf{R} \mB \mathbf{Q}) =  \dnnz(\mB)$, yielding $\Theta(\dnnz(\mB))$
complexity for the first triple product. The product $\mathbf{S} \cdot \mA$ only requires $\dlen(\mathsf{I})$ $\flops$ since it does not need to 
touch nonzeros of $\mA$ that do not contribute to $\mA(\mathsf{I},:)$. Similarly,
$(\mathbf{S} \mA) \cdot \mathbf{T}$ requires only $\dlen(\mathsf{J})$ $\flops$. The number of nonzeros in the second triple 
product is $\dnnz(\mathbf{S} \mA \mathbf{T}) = O(\dlen(\mathsf{I})+\dlen(\mathsf{J}))$. 
The final pointwise addition and subtraction (or generalized elementwise multiplication in the absence of additive inverses) operations take
time on the order of the total number of nonzeros in all operands~\cite{GALLA/sparse}, which is 
$O(\dnnz(\mA)+\dnnz(\mB)+\dlen(\mathsf{I})+\dlen(\mathsf{J}))$. 
\end{proof}

\subsection{SpRef in parallel}
\label{sec:parspref}

The parallelization of SpRef poses several challenges. 
The boolean matrices have only one nonzero per row or column.
For the parallel 2D algorithm to scale well with increasing number of processors, 
data structures and algorithms should respect hypersparsity~\cite{ipdps08}.
Communication should ideally take place along a single processor dimension, 
to save a factor of $\sqrt{p}$ in communication volume. As before, we assume a uniform 
distribution of nonzeros to processors in our analysis.

\begin{figure}
\begin{center}
\includegraphics[scale=0.55]{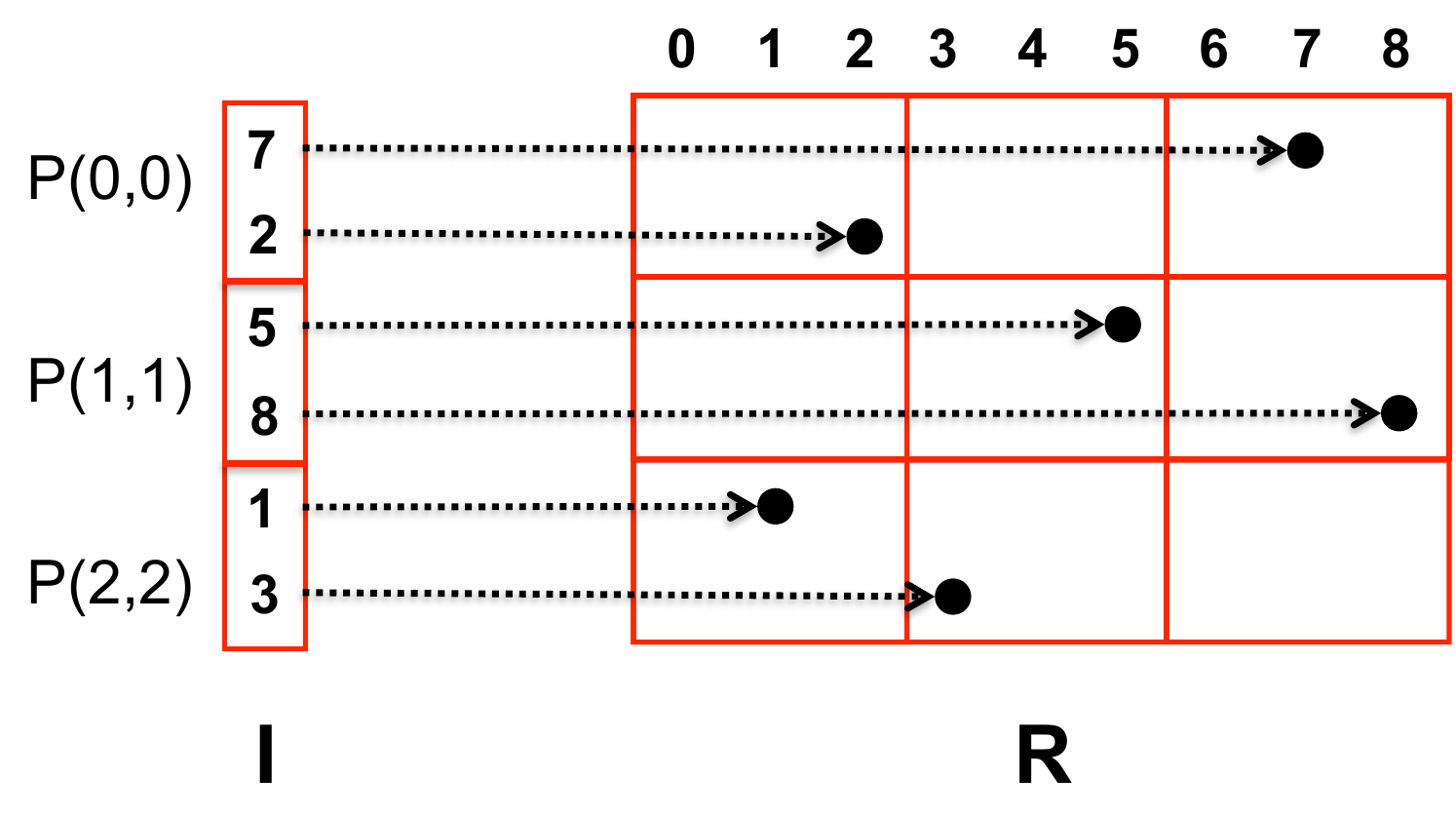}
\end{center}\caption{ Parallel forming of the left hand side boolean matrix $\mathbf{R}$ from the index vector $\mathsf{I}$ on 9 processors in 
a logical $3\times3$ grid. $\mathbf{R}$ will be subsequently multiplied 
with $\mA$ to extract 6 rows out of 9 from $\mA$ and order them as $\{7,2,5,8,1,3\}$.    
\label{fig:scatter}}
\end{figure} 

The communication cost of forming the $\mathbf{R}$ matrix in parallel is the cost of $\proc{Scatter}$ along the processor column. For the case of
vector $\mathsf{I}$ distributed to $\sqrt{p}$ diagonal processors, scattering can be implemented with an average communication cost of 
$\Theta(\alpha \cdot \lg{p} + \beta \cdot(\dlen(\mathsf{I})/\sqrt{p})$~\cite{ChanHPG07}. This process is illustrated in Figure~\ref{fig:scatter}.
The $\mathbf{Q}\transpose$ matrix can be constructed identically, followed by a $\proc{Transpose}(\mathbf{Q}\transpose)$ operation where
each processor $P(i,j)$ receives $\dnnz(\mathbf{Q})/p = \dlen(\mathsf{J})/p$ words of data from its diagonal neighbor $P(j,i)$. Note that the 
communication cost of the transposition is dominated by the cost of forming $\mathbf{Q}\transpose$ via $\proc{Scatter}$. 

While the analysis of our parallel SpRef algorithm assumes that the index vectors are distributed only on diagonal processors,
the asymptotic costs are identical in the 2D case where the vectors are distributed across all the processors~\cite{bfs:11}. This is because 
the number of elements (the amount of data) received by a given processor stays the same with the
only difference in the algorithm being the use of $\proc{Alltoall}$ operation instead of $\proc{Scatter}$ during the formation of the
$\mathbf{R}$ and $\mathbf{Q}$ matrices.

The parallel performance of SpGEMM is a complicated function of the matrix nonzero structures~\cite{icpp08,spgemm:10}. 
For SpRef, however, the special structure makes our analysis more precise. Suppose that the triple product is 
evaluated from left to right, $\mB=(\mathbf{R} \cdot \mA) \cdot \mathbf{Q}$; a similar analysis can be applied to the reverse evaluation. 
A conservative estimate of $\id{ni}(\mathbf{R},\mA)$, the number of indices $i$ for which $\mathbf{R}(:,i)\not= \emptyset$ 
and $\mA(i,:)\not= \emptyset$, is $\dnnz(\mathbf{R}) = \dlen(\mathsf{I})$. 

Using our \proc{HyperSparseGEMM}~\cite{GALLA/spgemm, ipdps08} as the computational kernel, time to compute the product 
$\mathbf{R} \mA$ (excluding the communication costs) is:

 \begin{equation*}
\label{eqn:tmult}
T_{mult} =  \underset{i,j}{\text{max}} \sum_{k=1}^{\sqrt{p}} \biggl(\!  \dnzc(\mathbf{R}_{ik})+\dnzr(\mA_{kj}) +  \flops(\mathbf{R}_{ik} \cdot \mA_{kj}) \cdot \lg{\id{ni}(\mathbf{R}_{ik},\mA_{kj})} \! \biggr),
\end{equation*}
where the maximum over all $(i,j)$ pairs is equal to the average, due to the uniform nonzero distribution assumption. 

Recall from the sequential analysis that 
 $\flops(\mathbf{R}\cdot \mA) \leq \dnnz(\mA)$ since each nonzero in $\mA$ contributes at most once to the overall flop count. We also know that 
 $\dnzc(\mathbf{R}) = \dlen(\mathsf{I})$ and $\dnzr(\mA) \leq \dnnz(\mA)$. Together with the uniformity 
 assumption, these identities yield the following results:
  \begin{align*}
 &\flops(\mathbf{R}_{ik} \cdot \mA_{kj}) = \frac{\dnnz(\mA)}{p\sqrt{p}}, \\
 &\id{ni}(\mathbf{R}_{ik} \cdot \mA_{kj}) \leq  \dnnz(\mathbf{R}_{ik}) = \frac{\dlen(\mathsf{I})}{p}, \\
 & \sum_{k=1}^{\sqrt{p}}  \dnzc(\mathbf{R}_{ik}) = \dnzc(\mathbf{R}_{i:}) =\frac{\dlen(\mathsf{I})}{\sqrt{p}},  \\
 &  \sum_{k=1}^{\sqrt{p}} \dnzr(\mA_{ik}) = \dnzr(\mA_{i:}) \leq \frac{\dnnz(\mA)}{\sqrt{p}}.
 \end{align*}

In addition to the multiplication costs, adding intermediate triples in $\sqrt{p}$ stages costs an extra 
$\flops(\mathbf{R}_{i:} \cdot \mA_{:j}) \lg{\sqrt{p}} = (\dnnz(\mA)/p)  \lg{\sqrt{p}} $ operations per processor.  
Thus, we have the following estimates of computation and 
communication costs for computing the product $\mathbf{R} \mA$: 

\begin{align*} 
T_{comp}(\mathbf{R}\cdot \mA) &=  O \Bigl( \frac{\dlen(\mathsf{I})  + \dnnz(\mA)}{\sqrt{p}} 
	+ \frac{\id{nnz}(\mA)}{p} \cdot \lg \bigl(\frac{\dlen(\mathsf{I})}{p} +\sqrt{p}\bigr) \Bigr), \\
T_{comm}(\mathbf{R}\cdot \mA) &= \Theta(\alpha \cdot \sqrt{p} + \beta \cdot \frac{\id{nnz}(\mA)}{\sqrt{p}}).
\end{align*}

Given that $\dnnz(\mathbf{R}\mA) \leq \dnnz(\mA)$, the analysis of multiplying the intermediate product $\mathbf{R} \mA$ with $\mathbf{Q}$ is similar.
Combined with the cost of forming auxiliary matrices $\mathbf{R}$ and $\mathbf{Q}$ and the costs of transposition of $\mathbf{Q}\transpose$, the total 
cost of the parallel SpRef algorithm becomes

\begin{align*} 
T_{comp} &=  O \Bigl( \frac{\dlen(\mathsf{I}) + \dlen(\mathsf{J})  + \dnnz(\mA)}{\sqrt{p}} 
	+ \frac{\id{nnz}(\mA)}{p} \cdot \lg \bigl(\frac{\dlen(\mathsf{I}) + \dlen(\mathsf{J}) }{p} +\sqrt{p}\bigr) \Bigr), \\
T_{comm} &= \Theta \Bigl ( \alpha \cdot \sqrt{p} + \beta \cdot \frac{\id{nnz}(\mA) + \dlen(\mathsf{I}) + \dlen(\mathsf{J})}{\sqrt{p}} \Bigr).
\end{align*}

We see that SpGEMM costs dominate the cost of SpRef. The asymptotic speedup is limited to $\Theta(\sqrt{p})$, as in the case of SpGEMM.

\section{Experimental Results}
\label{sec:expparallel}
We ran experiments on NERSC's Franklin system~\cite{franklin}, a 9660-node Cray XT4.
 Each XT4 node contains a quad-core 2.3~GHz AMD Opteron processor, attached
 to the XT4 interconnect via a Cray SeaStar2 ASIC using a  HyperTransport 2 
 interface capable of 6.4~GB/s. The SeaStar2 routing ASICs are connected in a 3D torus topology, 
 and each link is capable of 7.6~GB/s peak bidirectional bandwidth. Our algorithms perform similarly well on a fat tree topology, 
 as evidenced by our experimental results on the Ranger platform that are included in an earlier technical report~\cite{spgemm:10}.

We used the GNU C/\Cpp compilers (version 4.5), and Cray's MPI implementation, which is based on MPICH2.
We incorporated our code into the Combinatorial BLAS framework~\cite{combblas}. 
We experimented with core counts that are perfect squares, because the Combinatorial BLAS currently uses 
a square $\sqrt{p} \times \sqrt{p}$ processor grid. 
We compared performance with the Trilinos package (version 10.6.2.0)~\cite{Her2005},
which uses a 1D decomposition for its sparse matrices.  

In the majority of our experiments, we used synthetically generated \rmat\ matrices
rather than \erdosrenyi\ \cite{erdos} ``flat'' random matrices,
as these are more realistic for many graph analysis applications.  \rmat~\cite{rmat}, the Recursive MATrix generator,  
generates graphs with skewed degree distributions that approximate a power-law.  A scale $n$ \rmat\ matrix is $2^n$-by-$2^n$.
Our \rmat\ matrices have an average of $8$ nonzeros per row and column.
\rmat\ seed paratemeters are $a=.6$, and $b=c=d=.4/3$.
We applied a random symmetric permutation to the input matrices
to balance the memory and the computational load.
In other words, instead of storing and computing $\mC = \mA \mB$, 
we compute $\mathbf{P} \mC \mathbf{P}\transpose = (\mathbf{P} \mA \mathbf{P}\transpose )(\mathbf{P} \mB \mathbf{P}\transpose)$.
All of our experiments are performed on double-precision floating-point inputs.

Since algebraic multigrid on graphs coming from physical problems is an important case, 
we included two more matrices from the Florida Sparse Matrix collection~\cite{DavisH11} to our experimental analysis, 
into Section~\ref{sec:exp:c}, where we benchmark restriction operation that is used in algebraic multigrid. The first such matrix is 
a large circuit problem (Freescale1) with $17$ million nonzeros and $3.42$ million rows and columns. The second
matrix comes from a structural problem (GHS\_psdef/ldoor), and has $42.5$ million nonzeros and $952,203$ rows and columns. 
%We pruned explicit zero entries ($1.8$ million in Freescale1 and $4$ million in GHS\_psdef/ldoor) during preprocessing.

%We repeated some of the experiments on a machine with a different architecture, 
%namely the Ranger cluster~\cite{ranger} at TACC.
%Ranger has four 2.3GHz quad-core processors in each node, compared to one in Franklin. 
%The nodes are interconnected  through Infiniband with 1GB/sec unidirectional point-to-point bandwidth in a
%fat-tree topology. On Ranger, we experimented with multiple compilers and MPI implementations. 
%We report our best results, which we achieved using OpenMPI (v1.3b) and GNU C/C++ compilers (version 4.4). 

\subsection{Parallel Scaling of SpRef}
\label{sec:exp:spref}

\begin{figure*}
\centering
\begin{minipage}[t]{.47\textwidth}\centering
\includegraphics[scale=0.26]{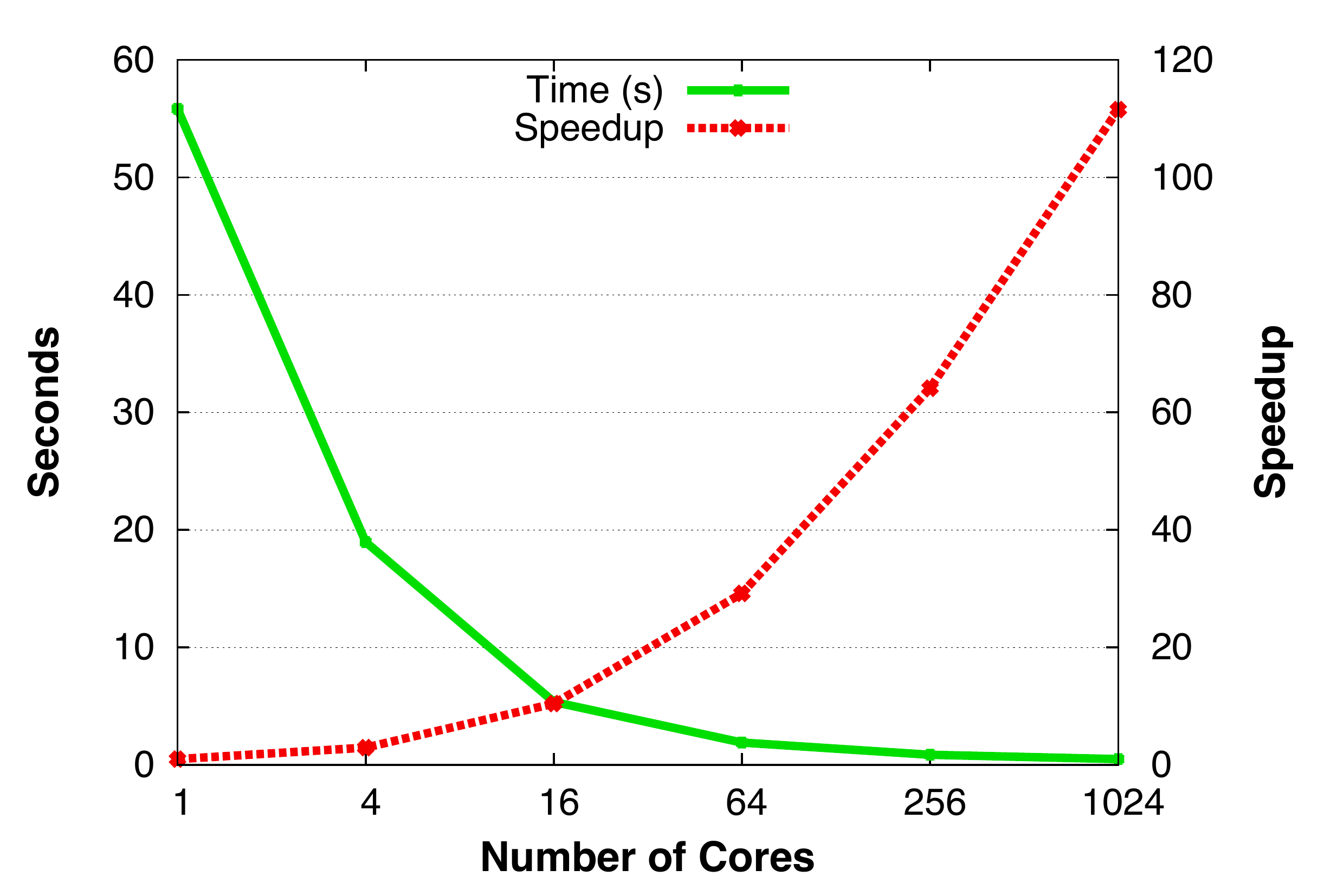}
\caption{Performance and parallel scaling of applying a random symmetric permutation to an \rmat\ matrix of scale 22. The x-axis uses a log scale.
 \label{fig:randpermperf}}
\end{minipage}
\hspace{0.2cm}
\begin{minipage}[t]{.47\textwidth}\centering
\includegraphics[scale=0.26]{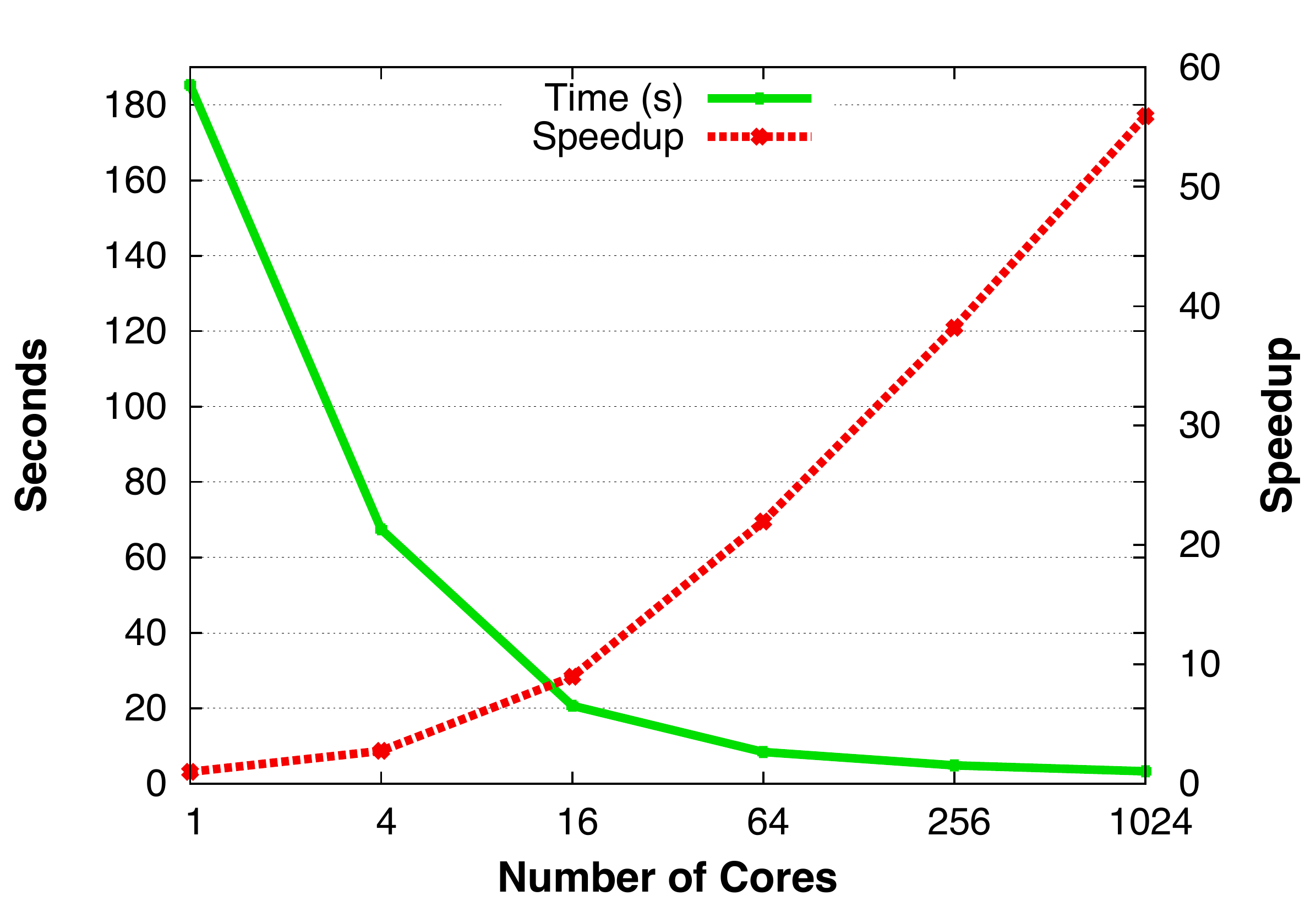}
\caption{Performance and parallel scaling of extracting 10 induced subgraphs from an \rmat\ matrix of scale 22. The x-axis uses a log scale. 
\label{fig:subgraphperf}}
\end{minipage}
\end{figure*}

Our first set of experiments randomly permutes the rows and columns of $\mA$, 
as an example case study for matrix reordering and partitioning. This operation corresponds to relabeling vertices of a graph.
Our second set of experiments explores subgraph extraction by generating 
a random permutation of $1:n$ and dividing it into $k \ll n$ chunks $r_1,\dots,r_k$. 
We then performed $k$ SpRef operations of the form $\mA(r_i,r_i)$, one after another 
(with a barrier in between). In both cases,  the sequential reference is our algorithm itself. 

The performance and parallel scaling of the symmetric random permutation is shown in Figure~\ref{fig:randpermperf}. 
The input is an \rmat\ matrix of scale 22 with approximately 32 million nonzeros in a square matrix of dimension $2^{22}$. 
 Speedup and runtime are plotted on different vertical axes. 
We see that scaling is close to linear up to about 64 processors, and proportional to $\sqrt{p}$ afterwards, 
agreeing with our analysis. 

The performance of subgraph extraction for $k=10$ induced subgraphs, each with $n/k$ randomly chosen vertices, 
is shown in Figure~\ref{fig:subgraphperf}. 
The algorithm performs well in this case too. 
The observed scaling is slightly less than the case of applying a single big permutation, 
which is to be expected since the multiple small subgraph extractions
increase span and decrease available parallelism.

\subsection{Parallel Scaling of SpAsgn}
\label{sec:exp:spasgn}

We benchmarked our parallel SpAsgn code by replacing a portion of the input matrix ($\mA$) with a structurally similar right-hand side matrix ($\mB$).
This operation is akin to replacing a portion of the graph due to a streaming update. The subset of vertices (row and column indices of $\mA$) to be updated  
is chosen randomly. In all the tests, the original graph is an \rmat\ matrix of scale 22 with 32 million nonzeros. The right-hand side (replacement) matrix 
is also an \rmat\ matrix of scales 21, 20, and 19, in three subsequent experiments, replacing 50\%, 25\%, and 12.5\% of the original graph, respectively.  
The average number of nonzeros per row and column are also adjusted for the right hand side matrices to match the nonzero density of the subgraphs they are replacing.

\begin{figure}
\centering
\includegraphics[scale=0.375]{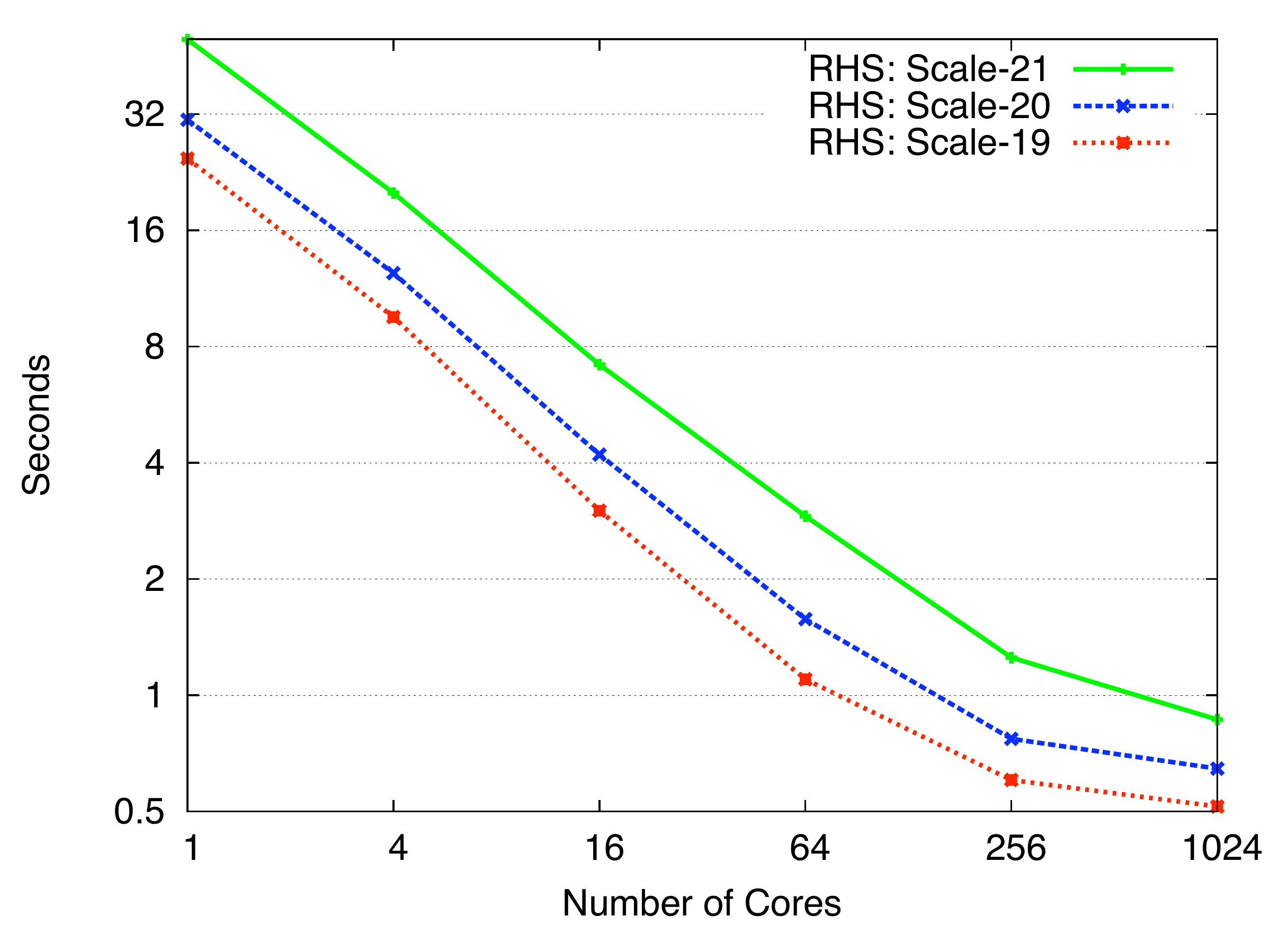} \vspace*{-2ex}
\caption[Observed speedup of of synchronous Sparse SUMMA]
{Observed scaling of the SpAsgn operation $\mA(\mathsf{I},\mathsf{I}) = \mB$ where $\mA \in \mathbb{R}^{\dimN \times \dimN}$ is an \rmat\ matrix of scale 22 and $\mB$ is another \rmat\ matrix whose 
scale is shown in the figure legend. $\mathsf{I}$ is a duplicate-free sequence with entries randomly selected from the range $1...n$; its length matches the dimensions of $\mB$. Both axes are log scale. \label{fig:spasgnexp}}
\end{figure}

The performance of this sparse matrix assignment operation is shown in Figure~\ref{fig:spasgnexp}. Our implementation uses a small number of Combinatorial BLAS routines: A sparse matrix constructor
from distributed vectors, essentially a parallel version of \matlab's {\tt sparse} routine, the generalized elementwise multiplication with direct support for negation,  and parallel SpGEMM implemented
using Sparse SUMMA. 

\subsection{Parallel Scaling of Sparse SUMMA} 
\label{sec:spgemmexp}

We implemented two versions of the 2D parallel SpGEMM algorithms in \Cpp using MPI. The first is directly based on 
Sparse SUMMA and is synchronous in nature, using blocking broadcasts. The second is asynchronous 
and uses one-sided communication in MPI-2. We found the asynchronous implementation to be consistently slower than
the broadcast-based synchronous implementation due to inefficient implementation of one-sided communication routines in MPI. 
Therefore, we only report the performance of the synchronous implementation. The motivation behind the asynchronous approach, performance comparisons, and implementation details, 
can be found in our technical report~\cite[Section 7]{spgemm:10}. On more than 4 cores of Franklin, synchronous implementation consistently outperformed the asynchronous  implementation
by 38-57\%. 

Our sequential \proc{HyperSparseGEMM} routines return a set of intermediate triples that are 
kept in memory up to a certain threshold without being merged immediately. This permits more balanced merging, 
eliminating some unnecessary scans that degraded performance in a preliminary implementation~\cite{icpp08}.

\subsubsection{Square Sparse Matrix Multiplication}
\label{sec:exp:a}

\begin{figure}
\centering
\includegraphics[scale=0.4]{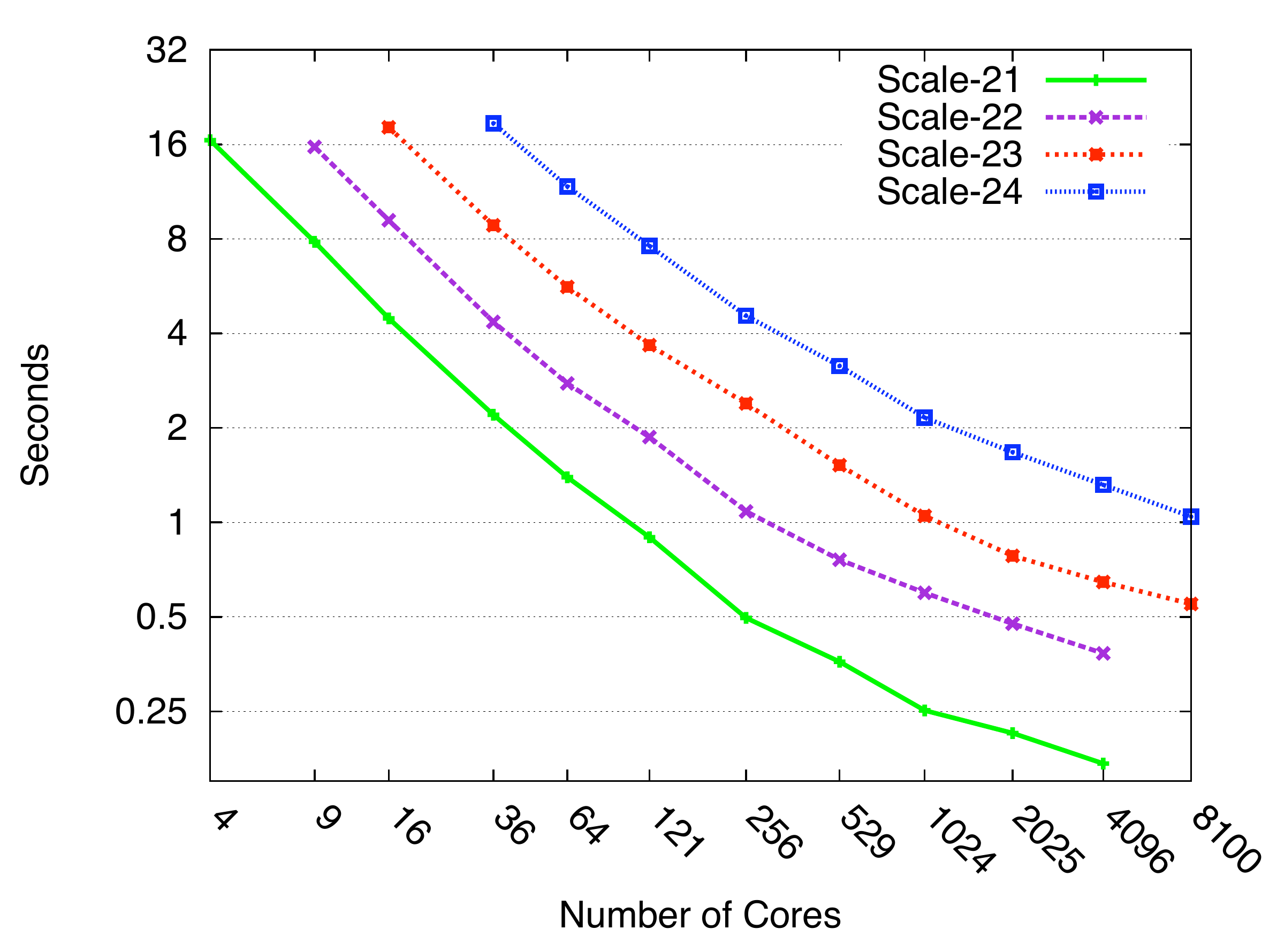} \vspace*{-2ex}
\caption[Observed speedup of of synchronous Sparse SUMMA]
{Observed scaling of synchronous Sparse SUMMA for the \rmat\ $\times$ \rmat\ product on matrices having dimensions $2^{21} - 2^{24}$.
Both axes are log scale. \label{fig:rmatrmatsynch}}
\end{figure}

In the first set of experiments, we multiply two structurally similar \rmat\ matrices. 
This square multiplication is representative of the expansion 
operation used in the Markov clustering algorithm~\cite{mclsimax}. 
It is also a challenging case for our implementation due to the highly skewed
nonzero distribution. 
We performed strong scaling experiments for matrix dimensions ranging from $2^{21}$ to $2^{24}$. 

Figure~\ref{fig:rmatrmatsynch} shows the speedup we achieved. The graph shows linear speedup until around 100 processors; afterwards the 
speedup is proportional to the square root of the number of processors. Both results agree with the theoretical analysis. To illustrate how the scaling 
transitions from linear to $\sqrt{p}$, we drew trend lines on the scale 21 results. As shown in Figure~\ref{fig:transition},  the slope of the log-log curve 
is $0.85$ (close to linear) until 121 cores, and the slope afterwards is $0.47$ (close to $\sqrt{p}$). Figure~\ref{fig:rmatsmall} 
zooms to the linear speedup regime, and shows the performance of our algorithm at lower concurrencies. The speedup and timings are plotted on 
different y-axes of the same graph. 

Our implementation of Sparse SUMMA achieves over 2 billion ``useful flops'' (in double precision)  per second on 8100 cores when multiplying scale 24 \rmat\ matrices. 
Since useful flops are highly dependent on the matrix structure and sparsity, we provide additional statistics for this operation in Table~\ref{tab:rmatstats}. 
Using matrices with more nonzeros per row and column will certainly yield higher performance rates (in useful flops). The gains from sparsity are clear if one considers 
dense flops that would be needed if these matrices were stored in a dense format. For example, multiplying two dense scale 24 matrices
requires 9444 exaflops.

\begin{figure*}
\centering
\begin{minipage}[t]{.48\textwidth}
\vspace{40pt}
\scalebox{0.82}{
\begin{tabular}{c|cccc}
Scale 		&  $\dnnz(\mA)$ 	&  $\dnnz(\mB)$   &  $\dnnz(\mC)$   	& $\flops$ \\
\hline
21	& 16.3 	& 16.3	& 123.9		& 253.2	\\
22	& 32.8	& 32.8 	& 257.1		& 523.7	\\
23	& 65.8	& 65.8	& 504.3		& 1021.3	\\
24	& 132.1 	& 132.1 	& 1056.9		& 2137.4 	\\
\end{tabular}}
\caption{Statistics about \rmat\ product $\mC = \mA \cdot \mB$. All numbers (except scale) are in millions.
\label{tab:rmatstats}}
\end{minipage}
\hspace{0.2cm}
\begin{minipage}[t]{.48\textwidth} \centering
\vspace{0pt}
\includegraphics[width=\textwidth]{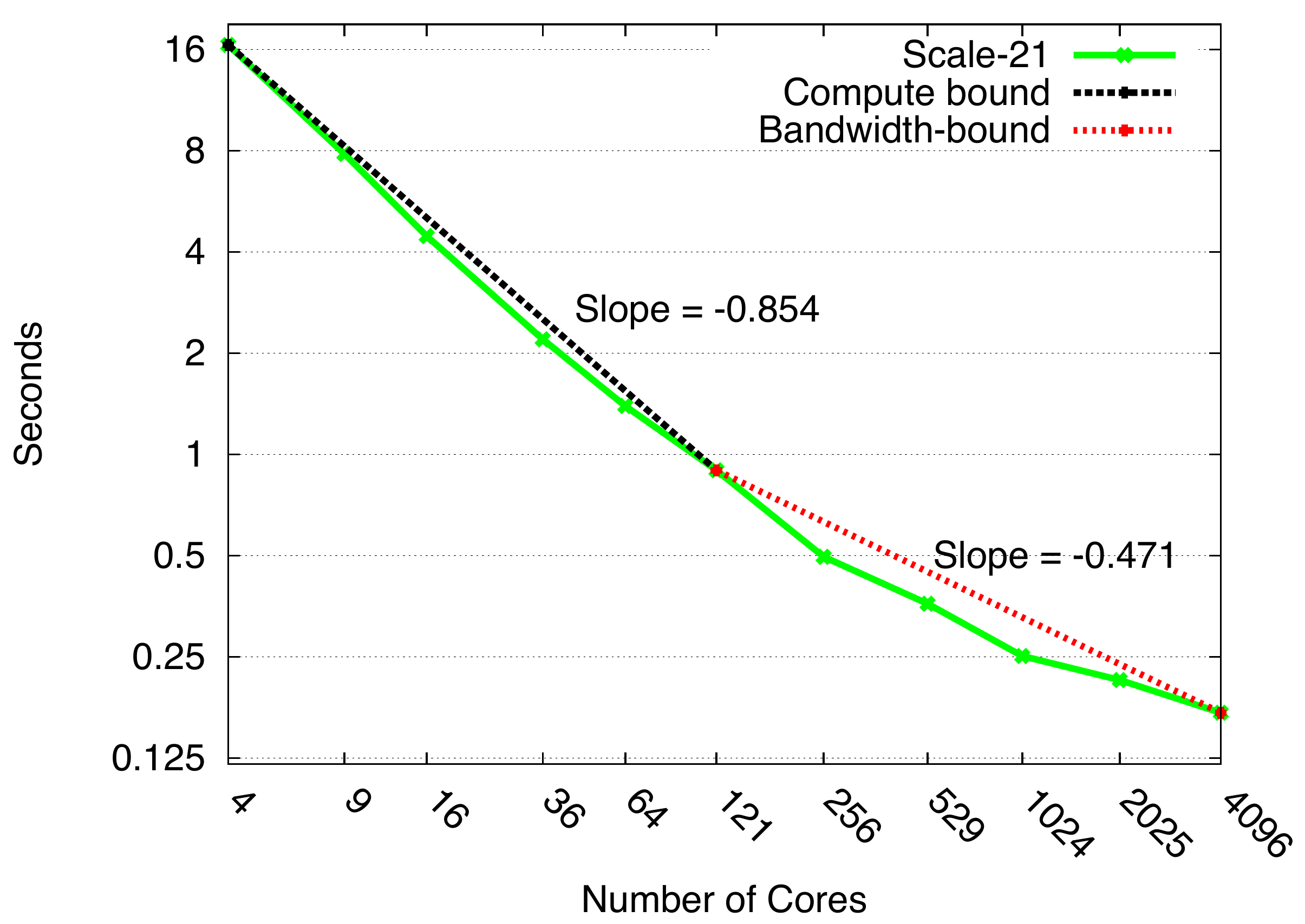}
\vspace{-12pt} 
\caption{Demonstration of two scaling regimes for scale 21 \rmat\ product.
 \label{fig:transition}}
\end{minipage}
\end{figure*}
 
 Figure~\ref{fig:commcomp} breaks down the time spent in communication and computation when multiplying
two \rmat\ graphs of scale 24. We see that computation scales much better than communication 
(over $90$x reduction when going from 36 to 8100 cores), implying that SpGEMM is communication bound for large concurrencies. 
For example, on 8100 cores, 83\% of the time is spent in communication. Communication times include the overheads due to 
synchronization and load imbalance.

Figure~\ref{fig:commcomp} also shows the effect of different blocking sizes. 
Remember that each processor owns a submatrix of 
size $n/\sqrt{p}$-by-$n/\sqrt{p}$.
On the left, the algorithm completes in $\sqrt{p}$ stages, each time broadcasting its whole local matrix.
On the right,  the algorithm completes in $2\sqrt{p}$ stages, each time broadcasting half of its local matrix. 
We see that while communication costs are not affected, the computation slows down by 1-6\% when doubling the number of stages.
This difference is due to the costs of splitting the input matrices before the multiplication and reassembling them afterwards, which
is small because splitting and reassembling are simple scans over the data whose costs are dominated by the cost of multiplication itself.

\begin{figure*}
\centering
\begin{minipage}[t]{.47\textwidth}\centering
\includegraphics[scale=0.26]{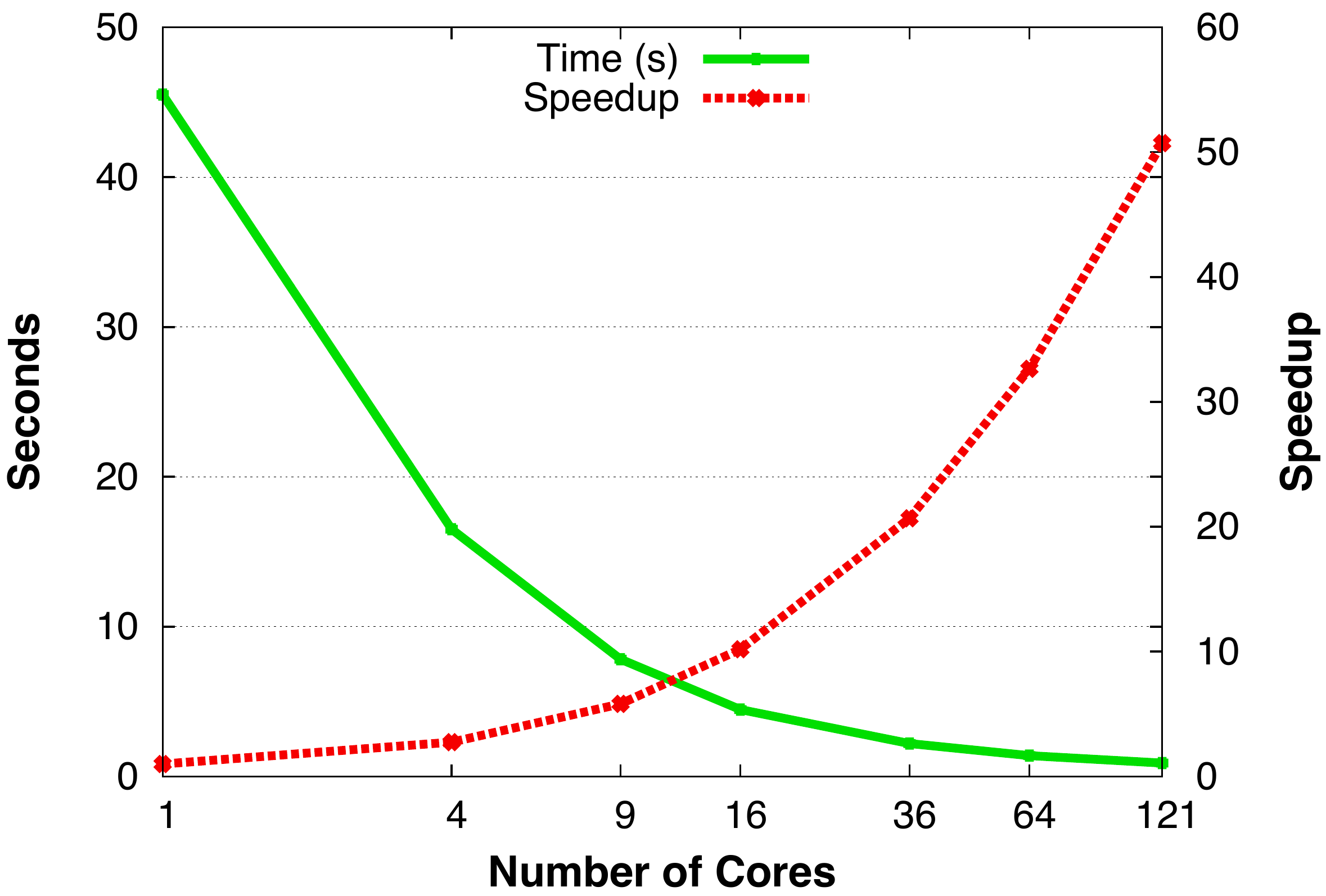}
\caption{Performance and scaling of Sparse SUMMA at lower concurrencies (scale 21 inputs). The x-axis uses a log scale.
 \label{fig:rmatsmall}}
\end{minipage}
\hspace{0.2cm}
\begin{minipage}[t]{.47\textwidth}\centering
\includegraphics[scale=0.55]{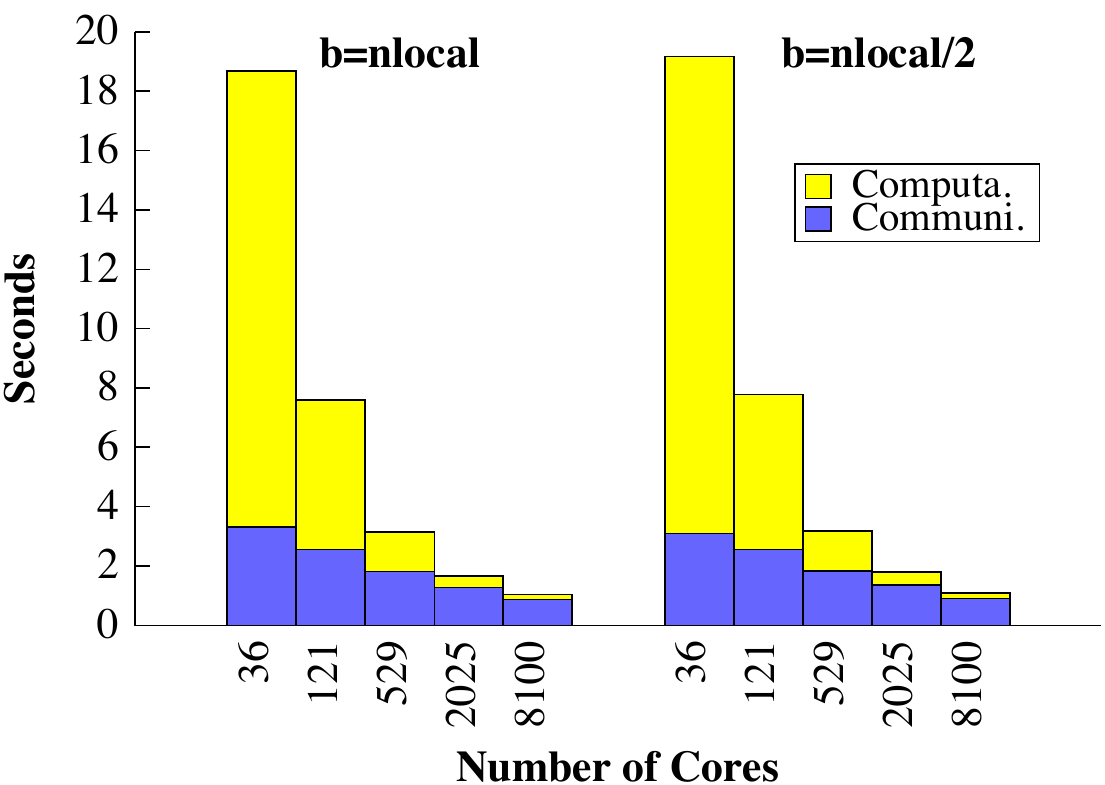}
\caption{Communication and computation breakdown, at various concurrencies and two blocking sizes (scale 24 inputs).  \label{fig:commcomp}}
\end{minipage}
\end{figure*}

\subsubsection{Multiplication with the Restriction Operator}
\label{sec:exp:c}

Multilevel methods are widely used in the solution of numerical and combinatorial problems~\cite{Teng99multilevel}. 
Such methods construct smaller problems by successive coarsening. 
The simplest coarsening is graph contraction: 
a contraction step chooses two or more vertices in the original graph $G$ to become a single 
aggregate vertex in the contracted graph $G'$. 
The edges of $G$ that used to be incident to any of the vertices forming the 
aggregate become incident to the new aggregate vertex in $G'$.
 
Constructing coarse grid during the V-cycle of algebraic multigrid~\cite{multigrid} or graph partitioning~\cite{brucegraph95} 
is a generalized graph contraction operation. Different algorithms need different coarsening operators. For example, a weighted 
aggregation~\cite{cedricilya09} might be preferred for partitioning problems. In general, coarsening can be represented 
as multiplication of the matrix representing the original fine domain (grid, graph, or hypergraph) by the restriction operator.

In these experiments, 
we use a simple restriction operation to perform graph contraction. 
Gilbert et al.~\cite{unifiedstarp} describe how to perform contraction using SpGEMM. 
Their algorithm creates a special sparse matrix $\mathbf{S}$ with $\dimN$ nonzeros. 
The triple product $\mathbf{S} \mA \mathbf{S}\transpose$ contracts the whole graph at once. Making $\mathbf{S}$ smaller in the first dimension while keeping the number of nonzeros same
changes the restriction order. For example, we contract the graph into half by using $\mathbf{S}$ having dimensions $\dimN /2 \times \dimN$, which is said to be of order 2.

\begin{figure*}
\centering
\begin{minipage}[t]{.48\textwidth}\centering
\includegraphics[scale=0.27]{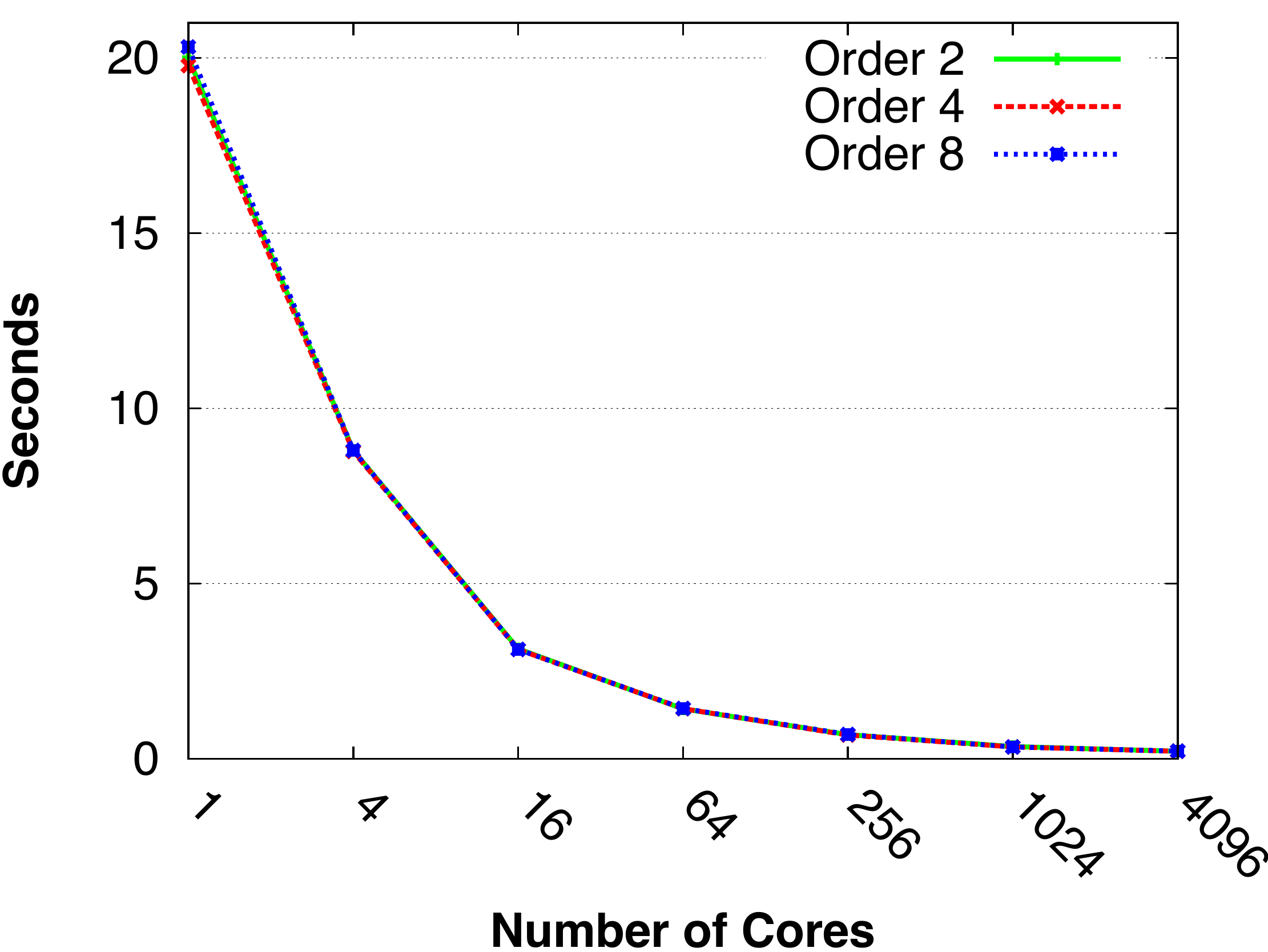}
\caption{Strong scaling of $\mB \gets \mA \mathbf{S}\transpose$, multiplying 
scale 23 \rmat\ matrices with the restriction operator on the right. 
The x-axis uses a log scale. \label{fig:galerkin}}
\end{minipage}
\hspace{0.2cm}
\begin{minipage}[t]{.48\textwidth}\centering
\includegraphics[scale=0.56]{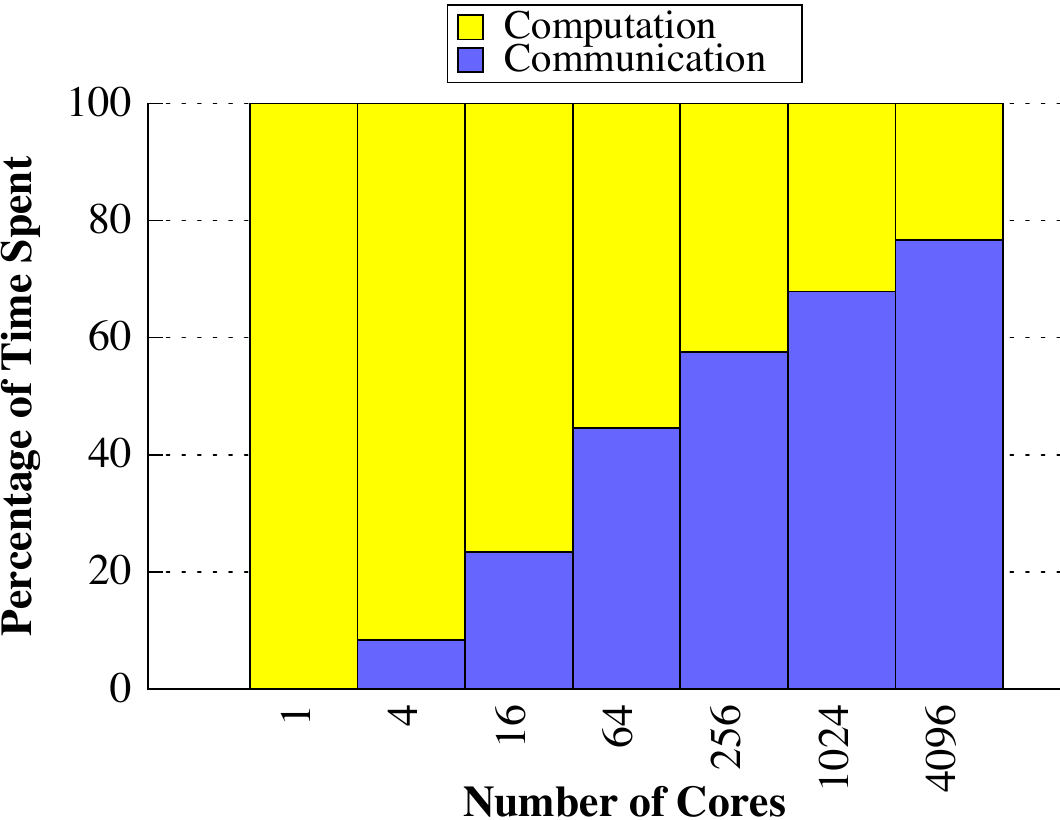}
\caption{Normalized communication and computation breakdown for multiplying 
scale 23 \rmat\ matrices with the restriction operator of order 4.  \label{fig:g_commcomp}}
\end{minipage}
\end{figure*}

Figure~\ref{fig:galerkin} shows `strong scaling' of $\mA \mathbf{S}\transpose$ operation for \rmat\ graphs of scale 23. 
We used restrictions of order 2, 4, and 8. Changing the interpolation order results in minor (less than 5\%) changes in performance, as shown by
the overlapping curves. This is further evidence that our algorithm's complexity is independent of the
matrix dimension, because interpolation order has a profound effect on the dimension of the right hand side matrix, but it does not change the expected $\flops$ and 
numbers of nonzeros in the inputs (it may slightly decrease the number of nonzeros in the output). The experiment shows scaling up to $4096$ processors. 
Figure~\ref{fig:g_commcomp} shows the breakdown of time (as percentages) spent on remote communication and local computation steps. 

Figures~\ref{fig:fullgalerkinleft} and~\ref{fig:fullgalerkinright} show `strong scaling' of the full restriction operation $\mathbf{S} \mA \mathbf{S}\transpose$ of order 8, 
using different parenthesizations for the triple product. The results show that our code achieves $110\times$ speedup on 1024-way concurrency and 
$163\times$ speedup on 4096-way concurrency, and the 
performance is not affected by the different parenthesizations.

\begin{figure}
\centering 
\subfloat[Left to right evaluation: $(\mathbf{S} \mA) \mathbf{S}\transpose$]{\label{fig:fullgalerkinleft} \includegraphics[width=0.49\textwidth]{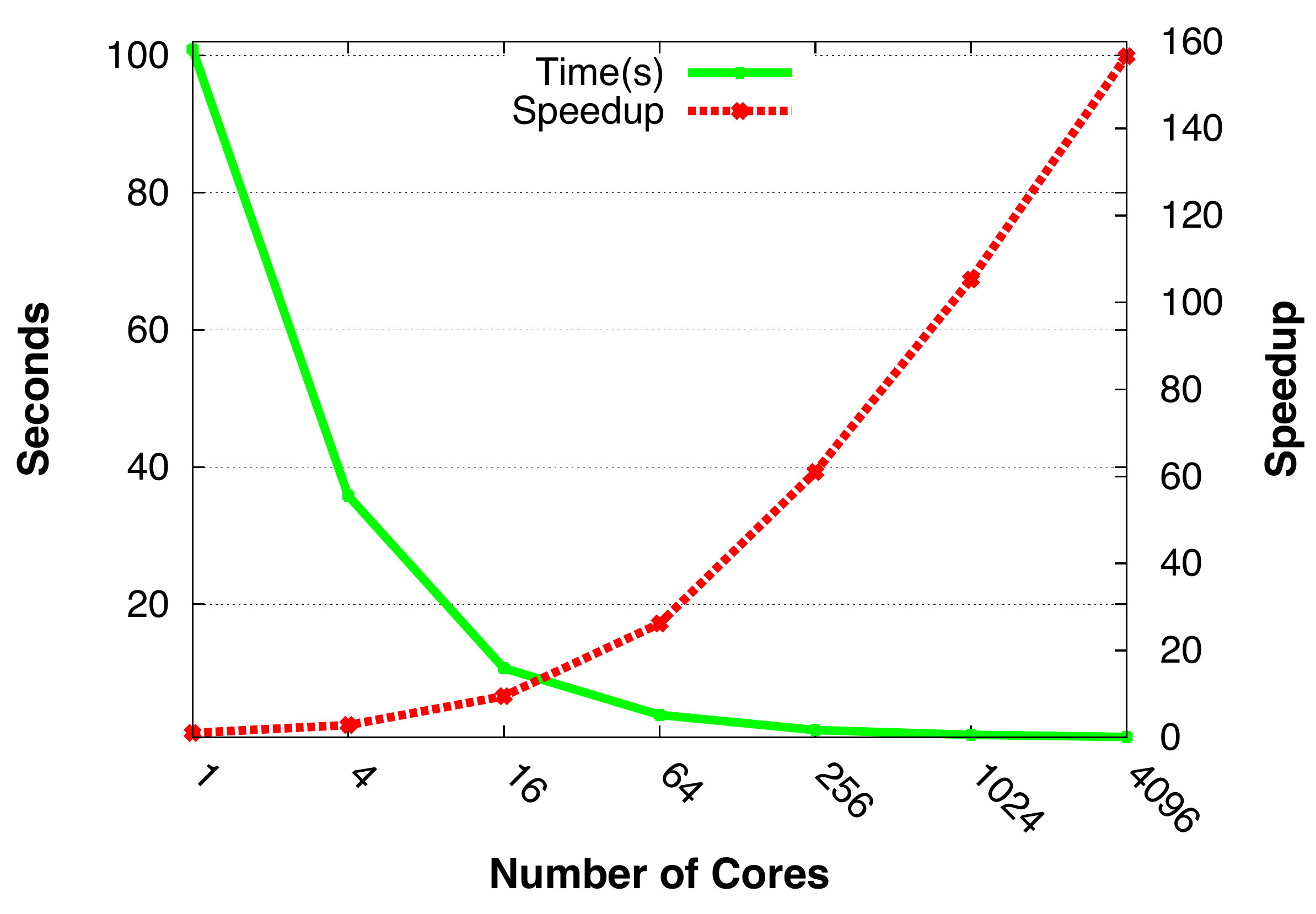}} 
\subfloat[Right to left evaluation: $\mathbf{S} (\mA \mathbf{S}\transpose)$]{\label{fig:fullgalerkinright} \includegraphics[width=0.49\textwidth]{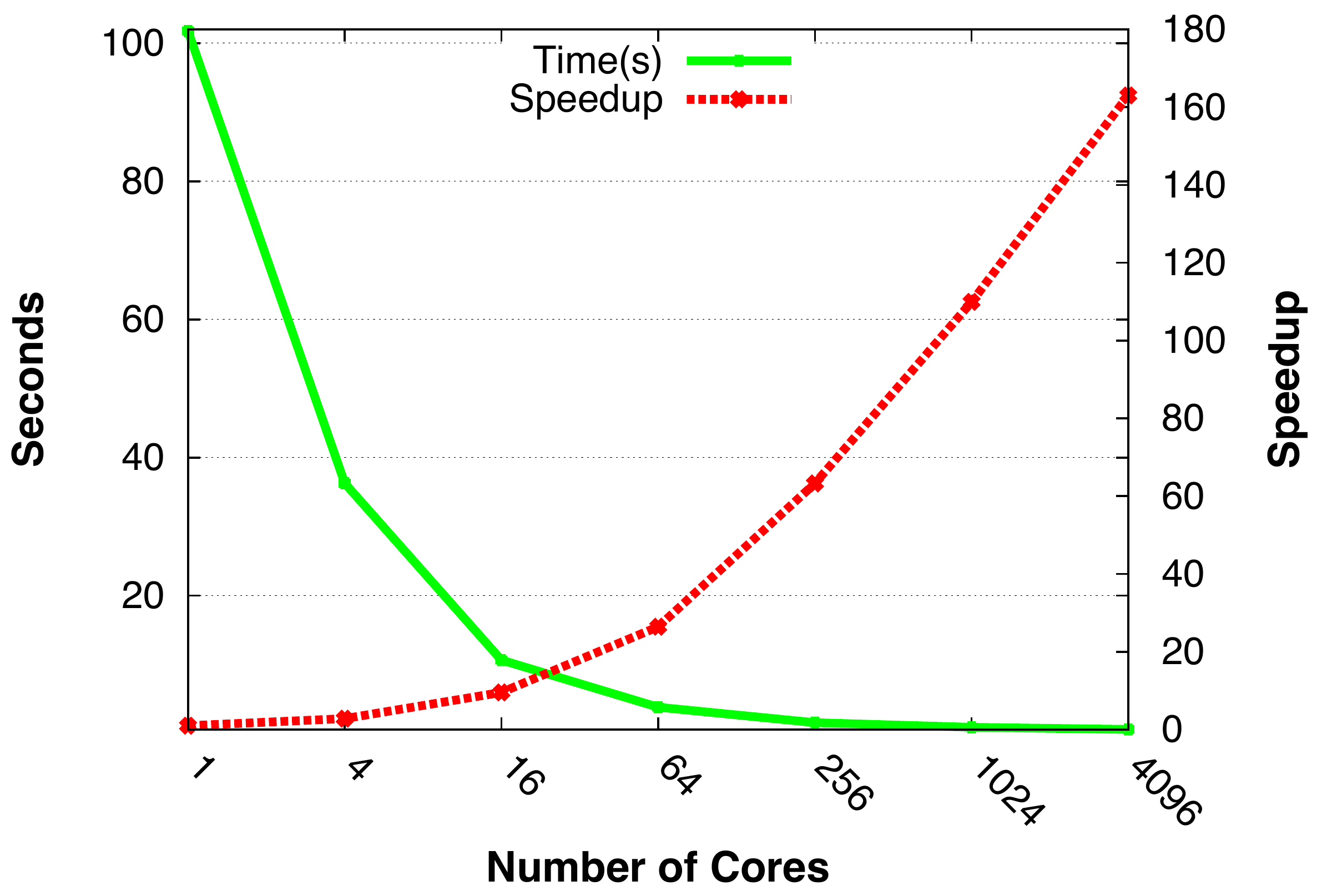}} 
\caption{The full restriction operation of order 8 applied to a scale 23 \rmat\ matrix. }  
\label{galerkintwosided}
\end{figure}

Figure~\ref{fig:restrictreal} shows the performance of full operation on real matrices from physical problems. Both matrices have a full diagonal that remains full after symmetric
permutation. Due to the 2D decomposition, processors responsible for the diagonal blocks typically have more work to do. 
For load-balancing and performance reasons, we split these matrices into two pieces $\mA = \mathbf{D}+\mathbf{L}$ where $\mathbf{D}$ is the diagonal piece and $\mathbf{L}$ is the off-diagonal piece. 
The restriction of rows becomes $\mathbf{S} \mA = \mathbf{S} \mathbf{D}  + \mathbf{S} \mathbf{L}$. Scaling the columns of $\mathbf{S}$  with the diagonal of $\mathbf{D}$ performs the former multiplication, and the latter multiplication uses Sparse SUMMA algorithm described in our paper. This splitting approach especially improved the scalability of restriction on Freescale1 matrix, because it is 
much sparser that GHS\_psdef/ldoor, which does not face severe load balancing issues.  Order 2 restriction shrinks the number of nonzeros from 17.0 to 15.3 million for Freescale1, and
from 42.5 to 42.0 million for GHS\_psdef/ldoor.

\begin{figure}[ht]
\centering 
\subfloat[Freescale1 matrix.]{\label{fig:restrict-freescale} \includegraphics[width=0.48\textwidth]{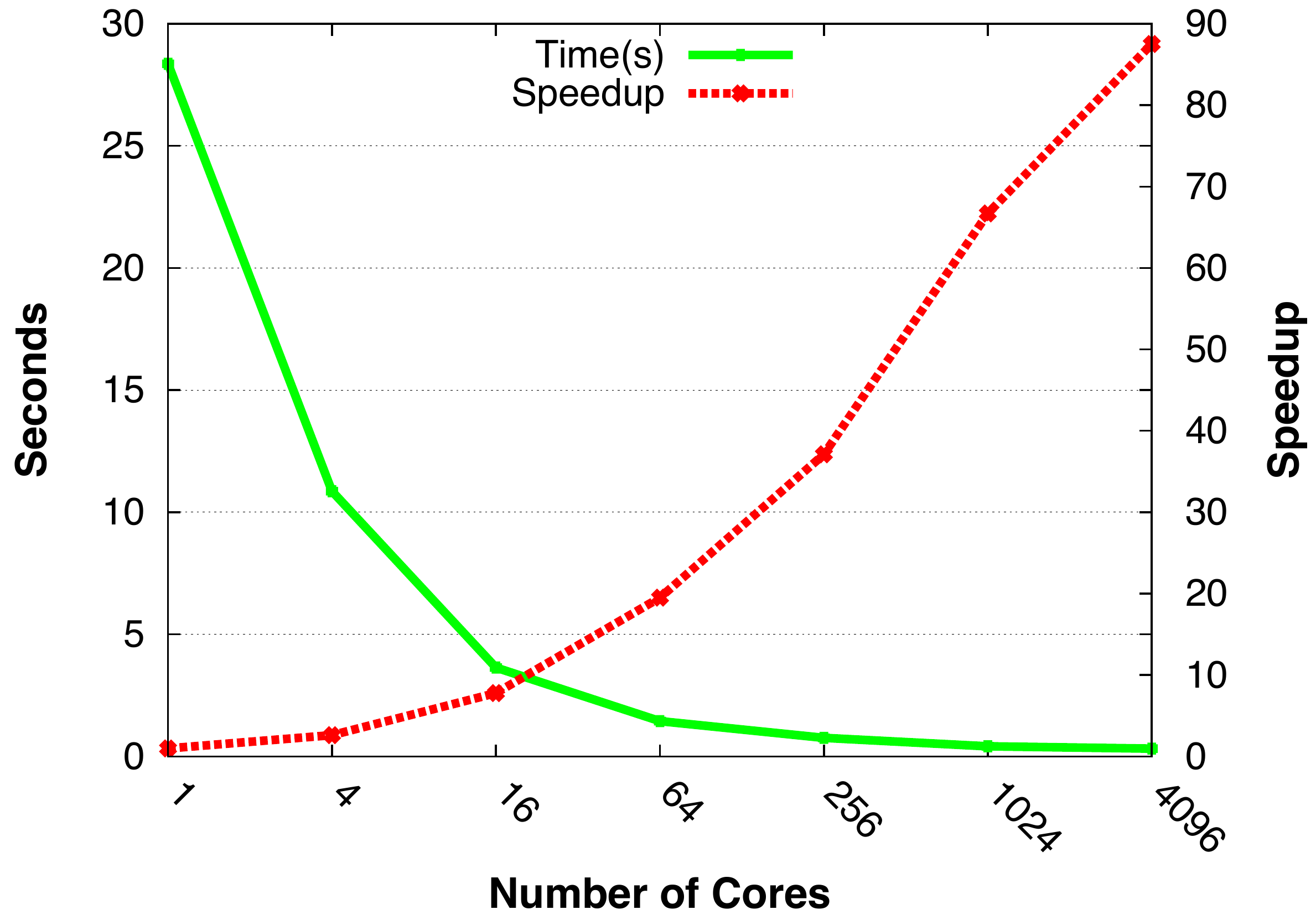}} 
\subfloat[GHS\_psdef/ldoor matrix.]{\label{fig:restrict-ldoor} \includegraphics[width=0.48\textwidth]{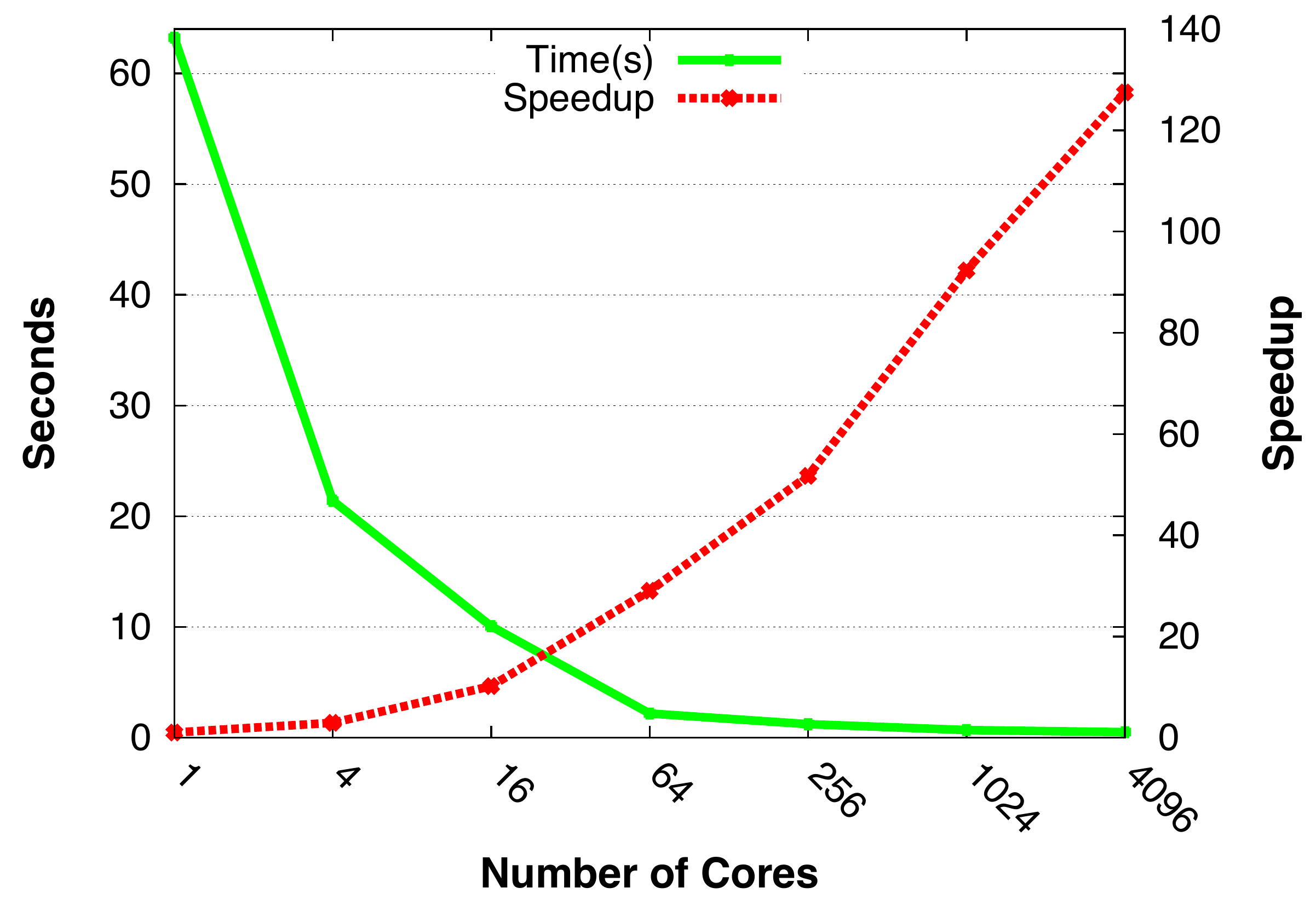}} 
\caption{Performance and strong scaling of Sparse SUMMA implementation for the full restriction of order 2 ($\mathbf{S} \mA \mathbf{S}\transpose$) on real matrices from physical problems. }  
\label{fig:restrictreal}
\end{figure}

\subsubsection{Tall Skinny Right Hand Side Matrix}
\label{sec:exp:b}

%\begin{figure}
%\centering
%\includegraphics[scale=0.4]{bfstrace} \vspace*{-2ex}
%\caption[Fringe size per level during breadth-first search]
%{Fringe size per level during breadth-first search. Each of the ten plots is an average of $256$ independent BFS operations on a graph of 1 million vertices
%and 8 million edges \label{fig:bfstrace}}
%\end{figure}

The last set of experiments multiplies \rmat\ matrices by tall skinny matrices of varying sparsity. 
This computation is representative of the parallel breadth-first search 
that lies at the heart of our distributed-memory betweenness centrality implementation~\cite{combblas}. 
This set indirectly examines the sensitivity to sparsity as well, 
because we vary the sparsity of the right hand side matrix
from approximately $1$ to $10^5$ nonzeros per column, in powers of 10. 
In this way, we imitate the patterns of the level-synchronous breadth-first search from multiple source vertices
where the current frontier can range from a few vertices to 
hundreds of thousands~\cite{bfs:11}.

For our experiments, the \rmat\ matrices on the left hand side have $d_1 = 8$ nonzeros 
per column and their dimensions vary from $\dimN=2^{20}$ to $\dimN=2^{26}$. 
The right-hand side is an \erdosrenyi\ matrix of dimensions $\dimN$-by-$\dimK$, and the number of
nonzeros per column, $d_2$, is varied from $1$ to $10^5$, in powers of $10$. 
The right-hand matrix's width $\dimK$ varies from $128$ to $8192$, growing proportionally to its length $\dimN$, hence keeping
the matrix aspect ratio constant at $n/k=8192$. Except for the $d_2 =10^5$ case, the \rmat\  matrix has more nonzeros than the right-hand matrix. 
In this computation, the total work is $W = O(d_1 d_2 \dimK)$, the total memory consumption is $M = O(d_1 \dimN + d_2 \dimK)$, 
and the total bandwidth requirement is $O(M \sqrt{p})$. 

We performed weak scaling experiments where memory consumption per processor is constant. 
Since $M = O(d_1 \dimN + d_2 \dimK)$, this is achieved by keeping both $\dimN / p = 2^{14}$ and $\dimK / p = 2$ constant. 
Work per processor is also constant. However, per-processor bandwidth requirements of this algorithm increases by a factor of $\sqrt{p}$. 

Figure~\ref{fig:weakskinny} shows a performance graph in three dimensions. 
The timings for each slice along the XZ-plane (i.e.\ for every 
$d_2=\{1,10,...,10^5\}$ contour) are normalized to the running time on 64 processors. 
We do not cross-compare the absolute performances for 
different $d_2$ values, as our focus in this section is parallel scaling. In line with the theory, we observe 
the expected $\sqrt{p}$ slowdown due to communication costs.  

The performance we achieved for these large scale experiments, where we ran our code on
up to $4096$ processors, is remarkable. It also shows that our implementation does not incur
any significant overheads since it does not deviate from the $\sqrt{p}$ curve. 

\begin{figure}
\centering
\includegraphics[scale=0.6]{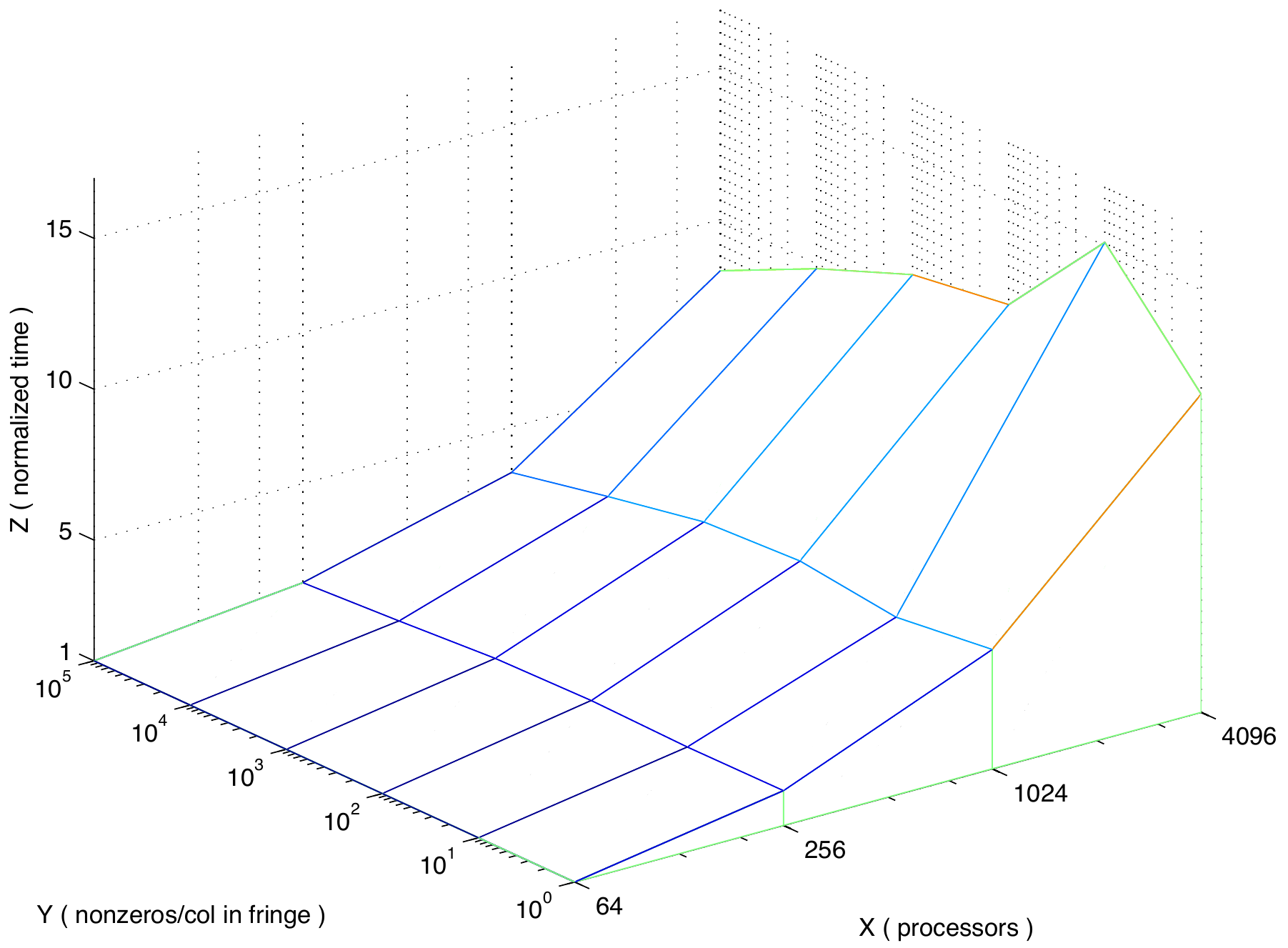} \vspace*{-2ex}
\caption[Weak scaling of \rmat\ times a tall skinny \erdosrenyi\ matrix]
{Weak scaling of \rmat\ times a tall skinny \erdosrenyi\ matrix. X (processors) and Y (nonzeros per column on fringe) axes are logarithmic, whereas Z 
(normalized time) axis is linear. \label{fig:weakskinny}}
\end{figure}

\subsection{Comparison with Trilinos}
The EpetraExt package of Trilinos can multiply two distributed sparse matrices in parallel. 
Trilinos can also permute matrices and extract submatrices through its Epetra\_Import and Epetra\_Export classes. These packages of Trilinos use a 1D data layout. 

For SpGEMM, we compared the performance of Trilinos's EpetraExt package with ours on two scenarios. In the first scenario, we multiplied two \rmat\
matrices as described in Section~\ref{sec:exp:a}, and in the second scenario, we multiplied an \rmat\ matrix with the restriction operator of order 8 on 
the right as described in Section~\ref{sec:exp:c}. 
 
Trilinos ran out of memory when multiplying \rmat\ matrices of scale larger than 21, or when using more than 256
processors. Figure~\ref{fig:gemm_epetra} shows SpGEMM timings for up to 256 processors on scale 21 data. 
Sparse SUMMA is consistently faster than Trilinos's implementation, with the gap increasing with the processor count, reaching $66\times$ on 256-way concurrency. 
Sparse SUMMA is also more memory efficient as Trilinos's matrix multiplication ran out of memory for $p =1$ and $p=4$ cores. The sweet spot for Trilinos seems to be around 120 cores, 
after which its performance degrades significantly. 

\begin{figure}[ht]
\centering 
\subfloat[\rmat\ $\times$ \rmat\ product (scale 21).]{\label{fig:gemm_epetra} \includegraphics[width=0.48\textwidth]{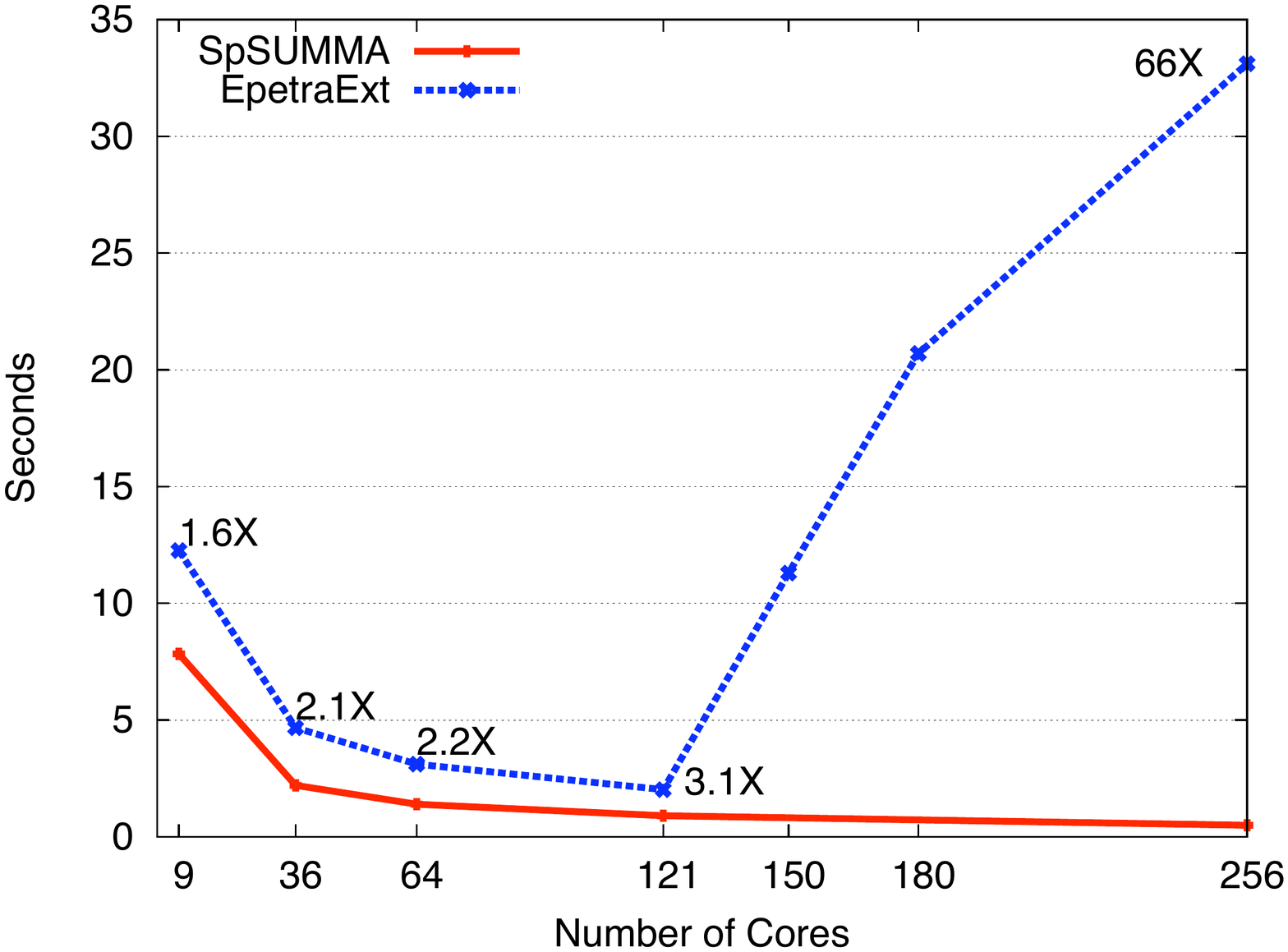}} 
\subfloat[Multiplication of an \rmat\ matrix of scale 23 with the restriction operator of order 8.]{\label{fig:galerkin_epetra} \includegraphics[width=0.48\textwidth]{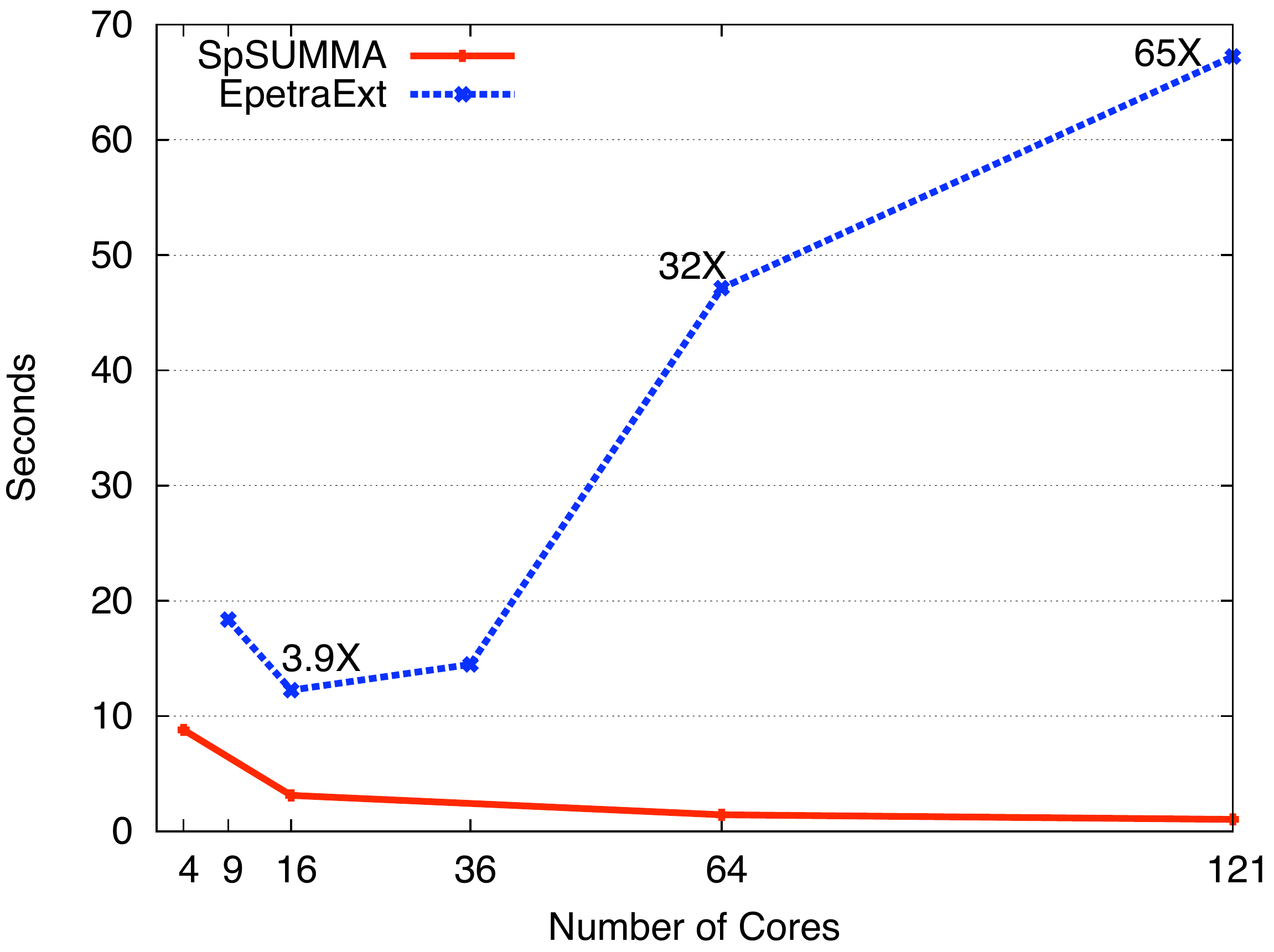}} 
\caption{Comparison of SpGEMM implementation of Trilinos's EpetraExt package with our Sparse SUMMA implementation using synthetically generated matrices. The data labels on
the plots show the speedup of Sparse SUMMA over EpetraExt.}  
\label{epetragemm}
\end{figure}

\begin{figure}[ht]
\centering 
\subfloat[Freescale matrix.]{\label{fig:freescale_epetra} \includegraphics[width=0.48\textwidth]{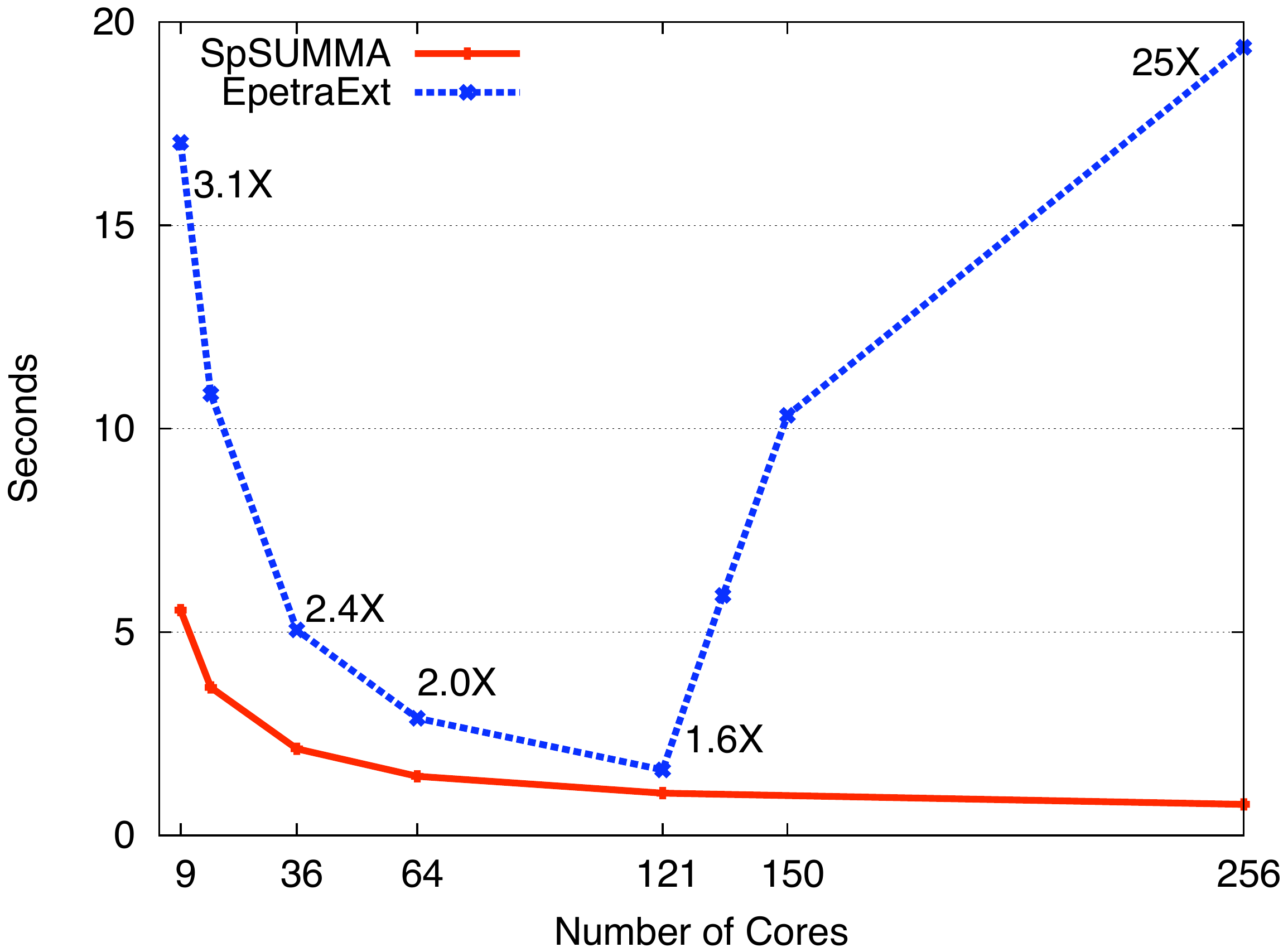}} 
\subfloat[GHS\_psdef/ldoor matrix.]{\label{fig:ldoor_epetra} \includegraphics[width=0.48\textwidth]{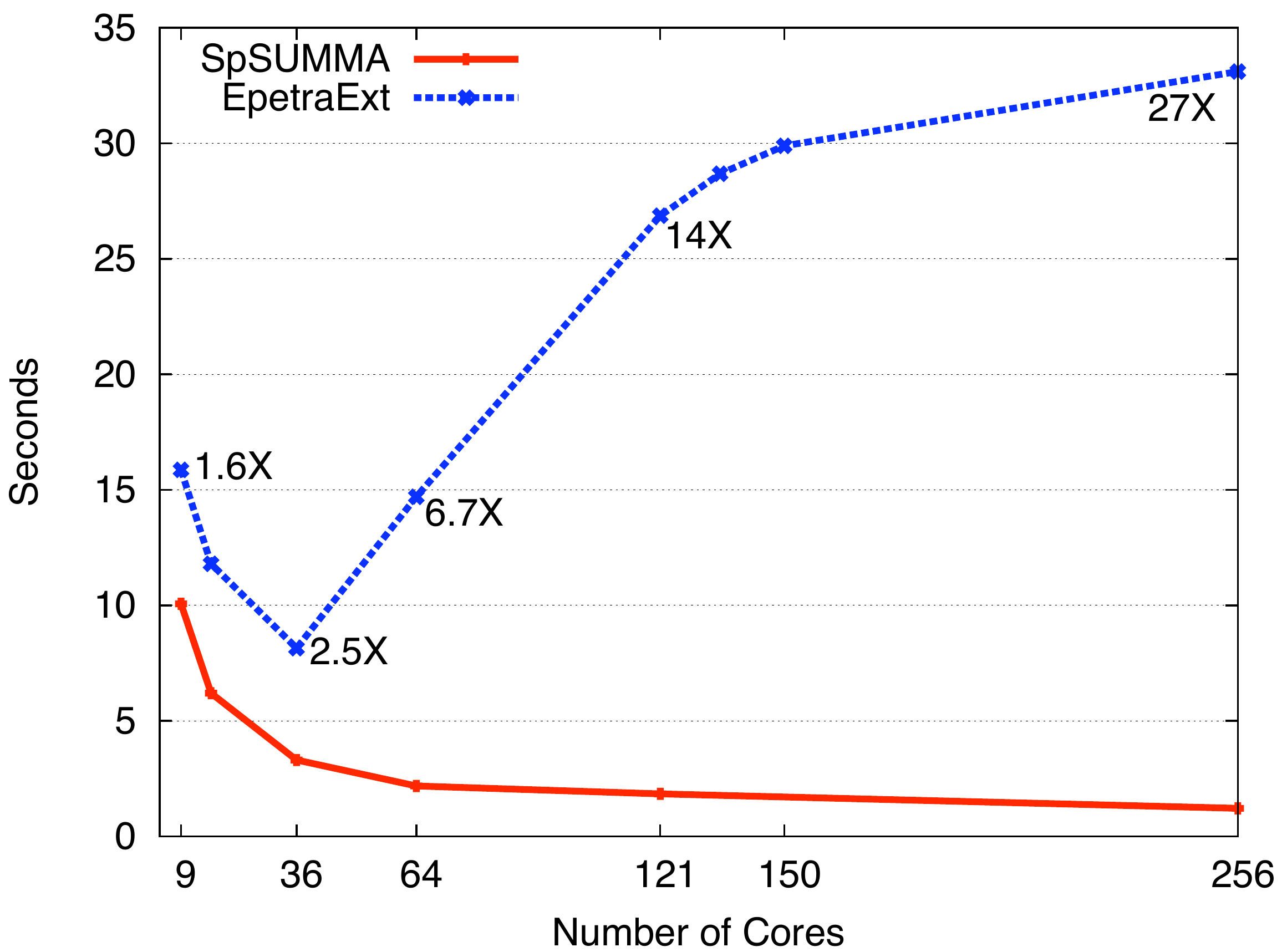}} 
\caption{Comparison of Trilinos's EpetraExt package with our Sparse SUMMA implementation for the full restriction of order 2 ($\mathbf{S} \mA \mathbf{S}\transpose$) on real matrices. The data labels on
the plots show the speedup of Sparse SUMMA over EpetraExt.}  
\label{epetrareal}
\end{figure}

In the case of multiplying with the restriction operator, the speed and scalability of our implementation over EpetraExt is even more pronounced. 
This is shown in Figure~\ref{fig:galerkin_epetra} where our code is 65X faster even on just 121 processors. Remarkably, our codes scales up to 4096
cores on this problem, as shown in Section~\ref{sec:exp:c}, while EpetraExt starts to slow down just beyond 16 cores. We also compared Sparse SUMMA with 
EpetraExt on matrices coming from physical problems, and the results for the full restriction operation ($\mathbf{S} \mA \mathbf{S}\transpose$) are shown in Figures~\ref{epetrareal}.

In order to benchmark Trilinos's sparse matrix indexing capabilities, we used EpetraExt's permutation class that can permute row or columns
of an Epetra\_CrsMatrix by creating a map defined by the permutation, followed by an Epetra\_Export operation to move data from the input object into the 
permuted object. We applied a random symmetric permutation on a \rmat\ matrix, as done in Section~\ref{sec:exp:spref}. 
Trilinos shows good scaling up to 121 cores but then it starts slowing down as concurrency increases, eventually becoming over 
$10\times$ slower than our SpRef implementation at 169 cores.

\section{Conclusions and Future Work}
We presented a flexible parallel sparse matrix-matrix multiplication (SpGEMM) algorithm, Sparse SUMMA, which scales to thousands of processors in distributed memory.
We used Sparse SUMMA as a building block to design and implement scalable parallel routines for sparse matrix indexing (SpRef) and assignment (SpAsgn). These operations 
are important in the context of graph operations. They yield elegant algorithms for coarsening graphs by edge contraction as in Figure~\ref{fig:contraction}, 
extracting subgraphs, performing parallel breadth-first search from multiple source vertices, and performing batch updates to a graph.  

We performed parallel complexity analyses of our primitives. In particular, using SpGEMM as a building block enabled the most general analysis of SpRef. 
Our extensive experiments confirmed that our implementation achieves the performance predicted by our analyses. 

\begin{figure}[t]
\centering  
\includegraphics[scale=0.45]{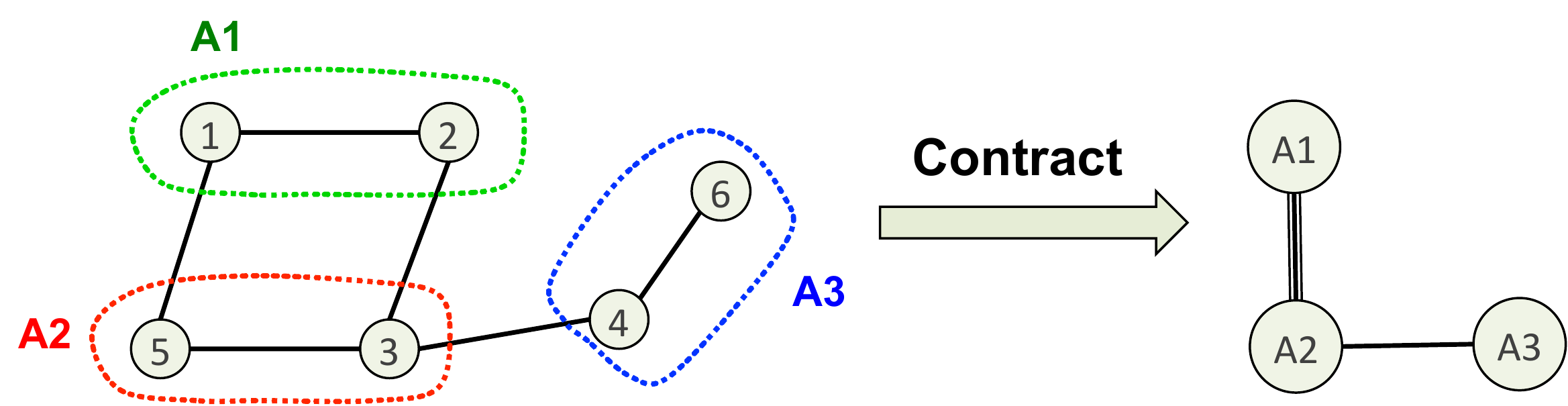}  
\caption{Example of graph coarsening using edge contraction, which can be implemented via a triple sparse matrix product $\mathbf{S} \mA \mathbf{S}\transpose$ 
where $\mathbf{S}$ is the restriction operator.} 
\label{fig:contraction}
\end{figure}

Our SpGEMM routine might be extended to handle matrix chain products. 
In particular, the sparse matrix triple product is used in the coarsening 
phase of algebraic multigrid~\cite{pamgdemmel}. 
Sparse matrix indexing and parallel graph contraction also require sparse matrix triple products~\cite{unifiedstarp}.
Providing a first-class primitive for sparse matrix chain products would eliminate temporary intermediate products and allow more optimization, 
such as performing structure prediction~\cite{cohen} 
and determining the best order of multiplication based on the sparsity structure of the matrices.

As we show in Section~\ref{sec:spgemmexp}, our implementation spends more than 75\% of its time in inter-node communication after 2000 processors.
Scaling to higher concurrencies require asymptotic reductions in communication volume. We are working on developing practical 
communication-avoiding algorithms~\cite{ca_demmel} for sparse matrix-matrix multiplication (and consequently for sparse matrix indexing and assignment), which might require inventing efficient novel sparse data structures to support such algorithms.

Our preliminary experiments suggest that synchronous algorithms for SpGEMM cause considerably higher load imbalance than
asynchronous ones~\cite[Section 7]{spgemm:10}. In particular, a truly one-sided implementation can perform up to 46\% faster 
when multiplying two \rmat\ matrices of scale 20 using 4000 processors. We will experiment with partitioned global address space (PGAS) languages, such as
UPC~\cite{upc}, because the current implementations of one-sided MPI-2 were not able to deliver satisfactory performance when used to implement 
asynchronous versions of our algorithms. 

As the number of cores per node increases due to multicore scaling, so does the contention on the network interface card.
Without hierarchical parallelism that exploits the faster on-chip network, the flat MPI parallelism will be unscalable 
because more processes will be competing for the same network link. 
Therefore, designing hierarchically parallel SpGEMM and SpRef algorithms is an important future direction.
 
\bibliographystyle{plain}
\bibliography{aydinthesis_jan11}

\begin{thebibliography}{10}

\bibitem{franklin}
{Franklin, Nersc's Cray XT4 System}.
\newblock \url{http://www.nersc.gov/users/computational-systems/franklin/}.

\bibitem{combblas_website}
{C}ombinatorial {BLAS} {L}ibrary ({MPI} reference implementation).
\newblock \url{http://gauss.cs.ucsb.edu/~aydin/CombBLAS/html/index.html}, 2012.

\bibitem{pamgdemmel}
Mark Adams and James~W. Demmel.
\newblock Parallel multigrid solver for 3d unstructured finite element
  problems.
\newblock In {\em Supercomputing '99: Proceedings of the 1999 ACM/IEEE
  conference on Supercomputing}, page~27, New York, NY, USA, 1999. ACM.

\bibitem{ca_demmel}
Grey Ballard, James Demmel, Olga Holtz, and Oded Schwartz.
\newblock Minimizing communication in numerical linear algebra.
\newblock {\em SIAM. J. Matrix Anal. \& Appl}, 32:pp. 866--901, 2011.

\bibitem{multigrid}
William~L. Briggs, Van~Emden Henson, and Steve~F. McCormick.
\newblock {\em A multigrid tutorial: second edition}.
\newblock Society for Industrial and Applied Mathematics, Philadelphia, PA,
  USA, 2000.

\bibitem{GALLA/spgemm}
Ayd{\i}n Bulu\c{c} and John~R. Gilbert.
\newblock New ideas in sparse matrix-matrix multiplication.
\newblock In J.~Kepner and J.~Gilbert, editors, {\em Graph Algorithms in the
  Language of Linear Algebra}. SIAM, Philadelphia.
\newblock 2011.

\bibitem{icpp08}
Ayd{\i}n Bulu\c{c} and John~R. Gilbert.
\newblock Challenges and advances in parallel sparse matrix-matrix
  multiplication.
\newblock In {\em ICPP'08: Proc. of the Intl. Conf. on Parallel Processing},
  pages 503--510, Portland, Oregon, USA, 2008. IEEE Computer Society.

\bibitem{ipdps08}
Ayd{\i}n Bulu\c{c} and John~R. Gilbert.
\newblock On the representation and multiplication of hypersparse matrices.
\newblock In {\em IPDPS'08: Proceedings of the 2008 IEEE International
  Symposium on Parallel\&Distributed Processing}, pages 1--11. IEEE Computer
  Society, 2008.

\bibitem{spgemm:10}
Ayd{\i}n Bulu\c{c} and John~R. Gilbert.
\newblock Highly parallel sparse matrix-matrix multiplication.
\newblock Technical Report UCSB-CS-2010-10, Computer Science Department,
  University of California, Santa Barbara, 2010.

\bibitem{combblas}
Ayd{\i}n Bulu\c{c} and John~R. Gilbert.
\newblock The {C}ombinatorial {BLAS}: Design, implementation, and applications.
\newblock {\em International Journal of High Performance Computing Applications
  (IJHPCA)}, 25(4):496--509, 2011.

\bibitem{GALLA/sparse}
Ayd{\i}n Bulu\c{c}, John~R. Gilbert, and Viral~B. Shah.
\newblock Implementing sparse matrices for graph algorithms.
\newblock In J.~Kepner and J.~Gilbert, editors, {\em Graph Algorithms in the
  Language of Linear Algebra}. SIAM, Philadelphia.
\newblock 2011.

\bibitem{bfs:11}
Ayd{\i}n Bulu\c{c} and Kamesh Madduri.
\newblock Parallel breadth-first search on distributed memory systems.
\newblock In {\em Proceedings of 2011 International Conference for High
  Performance Computing, Networking, Storage and Analysis}, SC '11, New York,
  NY, USA, 2011. ACM.

\bibitem{rmat}
Deepayan Chakrabarti, Yiping Zhan, and Christos Faloutsos.
\newblock {R-MAT}: A recursive model for graph mining.
\newblock In Michael~W. Berry, Umeshwar Dayal, Chandrika Kamath, and David~B.
  Skillicorn, editors, {\em SDM}. SIAM, 2004.

\bibitem{ChanHPG07}
Ernie Chan, Marcel Heimlich, Avi Purkayastha, and Robert~A. van~de Geijn.
\newblock Collective communication: theory, practice, and experience.
\newblock {\em Concurrency and Computation: Practice and Experience},
  19(13):1749--1783, 2007.

\bibitem{cedricilya09}
C\'{e}dric Chevalier and Ilya Safro.
\newblock Comparison of coarsening schemes for multilevel graph partitioning.
\newblock In {\em Learning and Intelligent Optimization: Third International
  Conference, LION 3. Selected Papers}, pages 191--205, Berlin, Heidelberg,
  2009. Springer-Verlag.

\bibitem{parallelblas}
Almadena Chtchelkanova, John Gunnels, Greg Morrow, James Overfelt, and
  Robert~A. van~de Geijn.
\newblock Parallel implementation of {BLAS}: {G}eneral techniques for {Level} 3
  {BLAS}.
\newblock {\em Concurrency: Prac\-tice and Experience}, 9(9):837--857, 1997.

\bibitem{cohen}
Edith Cohen.
\newblock Structure prediction and computation of sparse matrix products.
\newblock {\em Journal of Combinatorial Optimization}, 2(4):307--332, 1998.

\bibitem{upc}
{UPC} Consortium.
\newblock {UPC} language specifications, v1.2.
\newblock Technical Report LBNL-59208, Lawrence Berkeley National Laboratory,
  2005.

\bibitem{rkleene}
Paolo D'Alberto and Alexandru Nicolau.
\newblock R-{K}leene: A high-performance divide-and-conquer algorithm for the
  all-pair shortest path for densely connected networks.
\newblock {\em Algorithmica}, 47(2):203--213, 2007.

\bibitem{davisbook}
Timothy~A. Davis.
\newblock {\em Direct Methods for Sparse Linear Systems}.
\newblock Society for Industrial and Applied Mathematics, Philadelphia, PA,
  USA, 2006.

\bibitem{DavisH11}
Timothy~A. Davis and Yifan Hu.
\newblock The university of florida sparse matrix collection.
\newblock {\em ACM Trans. Math. Softw.}, 38(1):1, 2011.

\bibitem{erdos}
Paul Erd{\H o}s and Alfr\'{e}d R\'{e}nyi.
\newblock On random graphs.
\newblock {\em Publicationes Mathematicae}, 6(1):290--297, 1959.

\bibitem{summa}
R.~A. Van~De Geijn and J.~Watts.
\newblock {SUMMA}: {S}calable universal matrix multiplication algorithm.
\newblock {\em Concurrency: Prac\-tice and Experience}, 9(4):255--274, 1997.

\bibitem{smatlab}
John~R. Gilbert, Cleve Moler, and Robert Schreiber.
\newblock Sparse matrices in {M}atlab: {D}esign and implementation.
\newblock {\em SIAM Journal of Matrix Analysis and Applications},
  13(1):333--356, 1992.

\bibitem{unifiedstarp}
John~R. Gilbert, Steve Reinhardt, and Viral~B. Shah.
\newblock A unified framework for numerical and combinatorial computing.
\newblock {\em Computing in Science and Engineering}, 10(2):20--25, 2008.

\bibitem{gust:78}
Fred~G. Gustavson.
\newblock Two fast algorithms for sparse matrices: Multiplication and permuted
  transposition.
\newblock {\em ACM Transactions on Mathematical Software}, 4(3):250--269, 1978.

\bibitem{brucegraph95}
Bruce Hendrickson and Robert Leland.
\newblock A multilevel algorithm for partitioning graphs.
\newblock In {\em Supercomputing '95: Proceedings of the 1995 ACM/IEEE
  conference on Supercomputing}, page~28, New York, NY, USA, 1995. ACM.

\bibitem{Her2005}
Michael~A. Heroux, Roscoe~A. Bartlett, Vicki~E. Howle, Robert~J. Hoekstra,
  Jonathan~J. Hu, Tamara~G. Kolda, Richard~B. Lehoucq, Kevin~R. Long, Roger~P.
  Pawlowski, Eric~T. Phipps, Andrew~G. Salinger, Heidi~K. Thornquist, Ray~S.
  Tuminaro, James~M. Willenbring, Alan Williams, and Kendall~S. Stanley.
\newblock An overview of the {T}rilinos project.
\newblock {\em ACM Trans. Math. Softw.}, 31(3):397--423, 2005.

\bibitem{Kaplan:2006}
Haim Kaplan, Micha Sharir, and Elad Verbin.
\newblock Colored intersection searching via sparse rectangular matrix
  multiplication.
\newblock In {\em Proceedings of the twenty-second annual symposium on
  Computational geometry}, SCG '06, pages 52--60, New York, NY, USA, 2006. ACM.

\bibitem{multifrontal}
Joseph W.~H. Liu.
\newblock The multifrontal method for sparse matrix solution: Theory and
  practice.
\newblock {\em SIAM Review}, 34(1):pp. 82--109, 1992.

\bibitem{ogielski}
Andrew~T. Ogielski and William Aiello.
\newblock Sparse matrix computations on parallel processor arrays.
\newblock {\em SIAM Journal on Scientific Computing}, 14(3):519--530, 1993.

\bibitem{sparseclosure}
Gerald Penn.
\newblock Efficient transitive closure of sparse matrices over closed
  semirings.
\newblock {\em Theoretical Computer Science}, 354(1):72--81, 2006.

\bibitem{matching}
M.~O. Rabin and V.~V. Vazirani.
\newblock Maximum matchings in general graphs through randomization.
\newblock {\em Journal of Algorithms}, 10(4):557--567, 1989.

\bibitem{shahthesis}
Viral~B. Shah.
\newblock {\em An Interactive System for Combinatorial Scientific Computing
  with an Emphasis on Programmer Productivity}.
\newblock PhD thesis, University of California, Santa Barbara, June 2007.

\bibitem{Teng99multilevel}
Shang-Hua Teng.
\newblock Coarsening, sampling, and smoothing: Elements of the multilevel
  method.
\newblock In {\em Parallel Processing}, number 105 in The IMA Volumes in
  Mathematics and its Applications, pages 247--276, Germany, 1999.
  Springer-Verlag.

\bibitem{mclsimax}
Stijn {Van Dongen}.
\newblock Graph clustering via a discrete uncoupling process.
\newblock {\em SIAM Journal on Matrix Analysis and Applications},
  30(1):121--141, 2008.

\bibitem{YamazakiSchurComp10}
Ichitaro Yamazaki and Xiaoye Li.
\newblock On techniques to improve robustness and scalability of a parallel
  hybrid linear solver.
\newblock In {\em High Performance Computing for Computational Science –
  VECPAR 2010}, pages 421--434.

\bibitem{cycle}
Raphael Yuster and Uri Zwick.
\newblock Detecting short directed cycles using rectangular matrix
  multiplication and dynamic programming.
\newblock In {\em SODA '04: Proceedings of the fifteenth annual ACM-SIAM
  symposium on Discrete algorithms}, pages 254--260, 2004.

\end{thebibliography}
 
\end{document}